\documentclass[reqno,12pt]{amsart}

\usepackage{amsmath}
\usepackage{latexsym}
\usepackage{amssymb}
\usepackage{hyperref}
\usepackage{graphicx}
\usepackage{epstopdf}
\usepackage{ifpdf}
\ifpdf
\fi

\makeatletter
 
 \newcommand{\Rmnum}[1]{\expandafter\@slowromancap\romannumeral #1@}
 \makeatother

\newtheorem{theorem}{Theorem}[section]

\newtheorem{proposition}[theorem]{Proposition}
\newtheorem{lemma}[theorem]{Lemma}

\newtheorem{remark}[theorem]{Remark}
\newtheorem{assumption}[theorem]{Assumption}

\newcommand{\C}{{\mathbb C}}

\newcommand{\be}{\begin{equation}}
\newcommand{\ee}{\end{equation}}
\newcommand{\bea}{\begin{eqnarray}}
\newcommand{\eea}{\end{eqnarray}}
\newcommand{\ba}{\begin{array}}
\newcommand{\ea}{\end{array}}

\newcommand{\ol}{\overline}

\newcommand{\id}{\mathbb{I}}

\newcommand{\re}{\mathrm{Re}}

\newcommand{\eps}{\varepsilon}

\newcommand{\sig}{\sigma}
\newcommand{\Sig}{\Sigma}

\newcommand{\Lam}{\Lambda}
\newcommand{\gam}{\gamma}

\newcommand{\Om}{\Omega}

\newcommand{\dta}{\delta}
\newcommand{\Dta}{\Delta}

\newcommand{\tha}{\theta}

\newcommand{\pt}{\partial}

\numberwithin{equation}{section}

\begin{document}
\title[RHP for the 2-NLS equation on the interval]{Initial-boundary value problem for integrable nonlinear evolution equations with $3\times 3$ Lax pairs on the interval}

\author[J.Xu]{Jian Xu*}
\address{College of Science\\
University of Shanghai for Science and Technology\\
Shanghai 200093\\
People's  Republic of China}
\email{correspondence author: jianxu@usst.edu.cn}

\author[E.Fan]{Engui Fan}
\address{School of Mathematical Sciences, Institute of Mathematics and Key Laboratory of Mathematics for Nonlinear Science\\
Fudan University\\
Shanghai 200433\\
People's  Republic of China}
\email{faneg@fudan.edu.cn}


\date{\today}

\begin{abstract}
We present an approach for analyzing initial-boundary value problems which is formulated on the finite interval ($0\le x\le L$, where $L$ is a positive constant) for integrable equations whose Lax pairs involve $3\times 3$ matrices. Boundary value problems for integrable nonlinear evolution PDEs can be analyzed by the unified method introduced by Fokas and developed by him and his collaborators.
In this paper, we show that the solution can be expressed in terms of the solution of a $3\times 3$ Riemann-Hilbert problem. The relevant jump matrices are explicitly given in terms of the three matrix-value spectral functions $s(k)$,$S(k)$ and $S_L(k)$, which in turn are defined in terms of the initial values, boundary values at $x=0$ and boundary values at $x=L$, respectively. However, these spectral functions are not independent, they satisfy a global relation. Here, we show that the characterization of the unknown boundary values in terms of the given initial and boundary data is explicitly described for a nonlinear evolution PDE defined on the interval. Also, we show that in the limit when the length of the interval tends to infity, the relevant formulas reduce to the analogous formulas obtained for the case of boundary value problems formulated on the half-line.

\end{abstract}

\maketitle

\section{Introduction}

Integrable PDEs have the distinctive property that they can be written as the compatibility condition of two linear eigenvalue equations,
which are called a Lax pair \cite{lax}. An effective method (Inverse Scattering Transform (IST)) for solving the initial value problem
for integrable evolution equations on the line was discovered in 1967 \cite{ggkm}. However, the presence of a boundary presents new
challenges. It was realized in
\cite{f1} that the extension of this method to initial boundary value problems requires a
deeper understanding of the following question: What is the fundamental transform for
solving initial boundary value problems for linear evolution equations with $x$-derivatives
of arbitrary order? The investigation of this question has led to the discovery of
a general approach for solving boundary value problems for linear and for integrable
nonlinear PDEs \cite{f2} (see also \cite{f3,f4}). For integrable nonlinear evolution PDEs this approach is based on
the simultaneous spectral analysis of the two linear eigenvalue equations forming the
Lax pair, and on the investigation of the so-called global relation, which is an algebraic
relation coupling the relevant spectral functions.

\par
The Fokas method provides a generalization of the IST formalism from initial value to initial-boundary value (IBV) problems, and over the last eighteen years, this method has been used to analyze boundary value problems for several of the most important integrable equations with $2\times 2$ Lax pairs, such as the Korteweg-de Vries \cite{fi1}, the nonlinear Schr\"odinger \cite{fis}, the sine-Gordon equations \cite{fi2}, see \cite{fi3,fi4,abmfs1,abmfs2}. Just like the IST on the line, the unified method yields an expression for the solution of an IBV problem in terms of the solution of a Riemann-Hilbert problem. In particular, the asymptotic behavior of the solution can be analyzed in an effective way by using this Riemann-Hilbert problem and by employing the nonlinear version of the steepest descent method introduced by Deift and Zhou \cite{dz}.

\par
In 2012, Lenells first develops a methodology for analyzing IBV problems on the half-line for integrable evolution equations with Lax pairs involving $3\times 3$ matrices \cite{l1}. Although the transition from $2\times 2$ to $3\times 3$ matrix Lax pairs involves a number of novelties, the two
main steps of the method of \cite{f1,f4} remain the same:
\par
(1) Construct an integral representation of the solution characterized via a matrix
Riemann-Hilbert problem formulated on the complex $k$-sphere, where $k$ denotes the spectral parameter of the Lax pair. This representation involves,
in general, some unknown boundary values, thus the solution formula
is not yet effective.
\par
(2) Characterize the unknown boundary
values by analyzing the so-called global relation. In general, the
characterization of the unknown boundary values involves the solution
of a nonlinear problem.

\par
After Lenells' work, IBV problems on the half-line for other integrable evolution equations such as the Degasperis-Procesi \cite{l2}, Sasa-Satsuma \cite{jf1}, three wave \cite{jf2}, the two-component nonlinear Schr\"odinger \cite{jf3}, the Ostrovsky-Vakhnenko \cite{jf4} equations, are analyzed. However, within the knowledge of the authors, the IBV problems for integrable equations with $3\times 3$ matrices Lax pair on the finite interval has not been studied yet.

\par
The purpose of this paper is to extend the ideas from analyzing the IBV problems on the half-line to the finite interval for integrable evoultion equations with Lax pairs involving $3\times 3$ matrices. In fact, dealing with IBV problems on the interval has some difficulties. The implementation of step (1), we need four curve integration from the four corners of the $(x,t)-$domain. We will define analytic eigenfunctions, denoted by $\{M_n(x,t,k)\}$, via integral equations which involve integration from all {\em four} corners simultaneously. The most difficulties is to make a distinction between the integration contour $\gam_3$ and $\gam_4$ when we try to analyze the IBV problems on the interval. It is different from the analyzing the IBV problems on the half-line, because in the half-line case there just one integration curve $\gam_3$. Here, the constructions of this paper can be compared with the corresponding formalism for $2\times 2$-matrix Lax pairs introduced by Fokas and Its, see \cite{fi4}. The implementation of step (2), the differences are introducing a new factor $\frac{1}{\Dta}$ during analyzing the global relation to characterize the unknown boundary data in terms of the given initial and boundary data. We show that in the limit when the length of the interval tends to infity, the relevant formulas reduce to the analogous formulas obtained for the case of boundary value problems formulated on the half-line.
\par
\par
{\bf Organization of the paper:} In the next subsection, we introduce our main example considered in this paper. In section 2 we perform the spectral analysis of the associated Lax pair. We formulate the main Riemann-Hilbert problem in section 3 and this concludes
the implementation of step (1) above. We also get the map between the Dirichlet and the Neumann boundary problem through analyzing the global relation in section 4 and this concludes the implementation of step (2).

\subsection{The main example}
\par
In this paper, we will consider the two-component nonlinear Schr\"odinger equation or Manokov equation
\be\label{2-NLSe}
\left\{
\ba{l}
iq_{1t}+q_{1xx}-2\sig(|q_1|^2+|q_2|^2)q_1=0,\\
iq_{2t}+q_{2xx}-2\sig(|q_1|^2+|q_2|^2)q_2=0.
\ea
\right.\qquad \sig=\pm 1.
\ee
where $q_1(x,t)$ and $q_2(x,t)$ are complex-valued functions of $(x,t)\in \Omega$, with $\Omega$ denoting the finite interval domain
\be
\Omega=\{(x,t)|0\le x\le L, 0\le t\le T\},
\ee
here $L>0$ is a positive fixed constant and $T>0$ being a fixed final time. Here, $\sig=1$ means defocusing case and $\sig=-1$ means focusing case.
This system was first introduced by Manakov to describe the propagation of an optical
pulse in a birefringent optical fiber \cite{m74}. Subsequently, this system also arises
in the context of multicomponent Bose-Einstein condensates \cite{ba2001}.

\par
We will consider the following initial-boundary value problem for the 2-NLS equation,
\be\label{ibv-2nls}
\ba{lll}
\mbox{Initial value:}&q_{10}(x)=q_1(x,t=0),& q_{20}(x)=q_2(x,t=0),\\
\mbox{Dirichlet boundary value:}&g_{01}(t)=q_1(x=0,t),& g_{02}(t)=q_2(x=0,t),\\
&f_{01}(t)=q_1(x=L,t),& f_{02}(t)=q_2(x=L,t),\\
\mbox{Neumann boundary value:}&g_{11}(t)=q_{1x}(x=0,t),& g_{12}(t)=q_{2x}(x=0,t),\\
&f_{11}(t)=q_{1x}(x=L,t),& f_{12}(t)=q_{2x}(x=L,t).
\ea
\ee
It is well known that 2-NLS equation admits a $3\times 3$ Lax pair,
\begin{subequations}\label{Laxpair}
\be\label{Lax-x}
\Psi_x=U\Psi,\quad \Psi=\left(\ba{c}\Psi_1\\\Psi_2\\\Psi_3\ea\right).
\ee
\be\label{Lax-t}
\Psi_t=V\Psi.
\ee
\end{subequations}
where
\be\label{Udef}
U=ik\Lam+V_1.
\ee
and
\be\label{Vdef}
V=2ik^2\Lam+V_2
\ee
here
\be\label{Lamdef}
\Lam=\left(\ba{ccc}-1&0&0\\0&1&0\\0&0&1\ea\right),V_1=\left(\ba{ccc}0&q_1&q_2\\ \sig\bar q_1&0&0\\ \sig\bar q_2&0&0\ea\right),V_2=2kV_2^{(1)}+V_2^{(0)}.
\ee
where
\be
V_2^{(1)}=V_1,\qquad
V_2^{(0)}=i\Lam (V^2_1-V_{1x}).
\ee

\section{Spectral analysis}

\subsection{The closed one-form}

Introducing a new eigenfunction $\mu(x,t,k)$ by
\be\label{neweigfun}
\Psi=\mu e^{i\Lam kx+2i\Lam k^2t}
\ee
then we find the Lax pair equations
\be\label{muLaxe}
\left\{
\ba{l}
\mu_x-[ik\Lam,\mu]=V_1\mu,\\
\mu_t-[2ik^2\Lam,\mu]=V_2\mu.
\ea
\right.
\ee
Letting $\hat A$ denotes the operators which acts on a $3\times 3$ matrix $X$ by $\hat A X=[A,X]$ , then the equations in (\ref{muLaxe}) can be written in differential form as
\be\label{mudiffform}
d(e^{-(ikx+2ik^2t)\hat \Lam}\mu)=W,
\ee
where $W(x,t,k)$ is the closed one-form defined by
\be\label{Wdef}
W=e^{-(ikx+2ik^2t)\hat \Lam}(V_1dx+V_2dt)\mu.
\ee

\subsection{The $\mu_j$'s definition}
We define four eigenfunctions $\{\mu_j\}_1^4$ of (\ref{muLaxe}) by the Volterra integral equations
\be\label{mujdef}
\mu_j(x,t,k)=\id+\int_{\gam_j}e^{(i kx+2i k^2t)\hat \Lam}W_j(x',t',k).\qquad j=1,2,3,4.
\ee
where $W_j$ is given by (\ref{Wdef}) with $\mu$ replaced with $\mu_j$, and the contours $\{\gam_j\}_1^4$ are showed in Figure 1.
\begin{figure}[th]
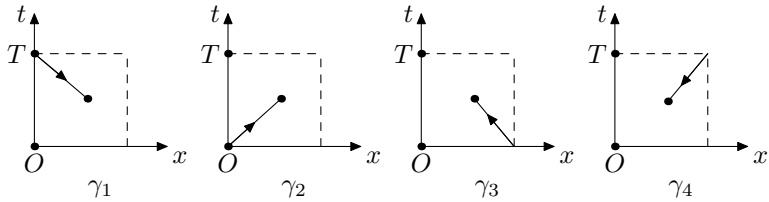

\centering
\includegraphics{2-NLS-INT.1}
\includegraphics{2-NLS-INT.2}
\includegraphics{2-NLS-INT.3}
\includegraphics{2-NLS-INT.4}
\caption{The three contours $\gam_1,\gam_2,\gam_3$ and $\gam_4$ in the $(x,t)-$domain.}
\end{figure}
The first, second and third column of the matrix equation (\ref{mujdef}) involves the exponentials
\be
\ba{ll}
\mbox{$[\mu_j]_1$:}&e^{2ik(x-x')+4ik^2(t-t')},e^{2ik(x-x')+4ik^2(t-t')}\\
\mbox{$[\mu_j]_2$:}&e^{-2ik(x-x')-4ik^2(t-t')},\\
\mbox{$[\mu_j]_3$:}&e^{-2ik(x-x')-4ik^2(t-t')}.
\ea
\ee
And we have the following inequalities on the contours:
\be
\ba{ll}
\gam_1:&x-x'\ge 0,t-t'\le 0,\\
\gam_2:&x-x'\ge 0,t-t'\ge 0,\\
\gam_3:&x-x'\le 0,t-t'\ge 0,\\
\gam_4:&x-x'\le 0,t-t'\le 0.
\ea
\ee
So, these inequalities imply that the functions $\{\mu_j\}_1^4$ are bounded and analytic for $k\in\C$ such that $k$ belongs to
\be\label{mujbodanydom}
\ba{ll}
\mu_1:&(D_2,D_3,D_3),\\
\mu_2:&(D_1,D_4,D_4),\\
\mu_3:&(D_3,D_2,D_2),\\
\mu_4:&(D_4,D_1,D_1).
\ea
\ee
where $\{D_n\}_1^4$ denote four open, pairwisely disjoint subsets of the complex $k-$sphere showed in Figure 2.
\begin{figure}[th]
\centering
\includegraphics{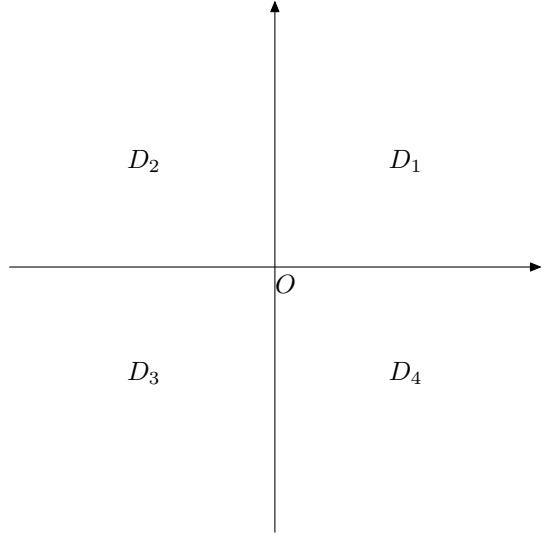}
\caption{The sets $D_n$, $n=1,\ldots ,4$, which decompose the complex $k-$plane.}\label{fig2}
\end{figure}
\par
We also notice that the sets $\{D_n\}_1^4$ has the following properties:
\[
\ba{l}
D_1=\{k\in\C|\re{l_1}>\re{l_2}=\re{l_3},\re{z_1}>\re{z_2}=\re{z_3}\},\\
D_2=\{k\in\C|\re{l_1}>\re{l_2}=\re{l_3},\re{z_1}<\re{z_2}=\re{z_3}\},\\
D_3=\{k\in\C|\re{l_1}<\re{l_2}=\re{l_3},\re{z_1}>\re{z_2}=\re{z_3}\},\\
D_4=\{k\in\C|\re{l_1}<\re{l_2}<\re{l_3},\re{z_1}<\re{z_2}=\re{z_3}\},\\
\ea
\]
where $l_i(k)$ and $z_i(k)$ are the diagonal entries of matrices $ik\Lam$ and $2ik^2\Lam$, respectively.
\par

\subsection{The $M_n$'s definition}
For each $n=1,\ldots,4$, define a solution $M_n(x,t,k)$ of (\ref{muLaxe}) by the following system of integral equations:
\be\label{Mndef}
(M_n)_{ij}(x,t,k)=\dta_{ij}+\int_{\gam_{ij}^n}(e^{(i kx+2i k^2t)\hat \Lam}W_n(x',t',k))_{ij},\quad k\in D_n,\quad i,j=1,2,3.
\ee
where $W_n$ is given by (\ref{Wdef}) with $\mu$ replaced with $M_n$, and the contours $\gam_{ij}^n$, $n=1,\ldots,4$, $i,j=1,2,3$ are defined by
\be\label{gamijndef}
\gam_{ij}^n=\left\{
\ba{lclcl}
\gam_1&if&\re l_i(k)<\re l_j(k)&and&\re z_i(k)\ge\re z_j(k),\\
\gam_2&if&\re l_i(k)<\re l_j(k)&and&\re z_i(k)<\re z_j(k),\\
\gam_3&if&\re l_i(k)\ge\re l_j(k)&and&\re z_i(k)\le \re z_j(k),\\
\gam_4&if&\re l_i(k)\ge\re l_j(k)&and&\re z_i(k)\ge \re z_j(k).
\ea
\right.
\quad \mbox{for }\quad k\in D_n.
\ee
Here, we make a distinction between the contours $\gam_3$ and $\gam_4$ as follows,
\be
\gam^{n}_{ij}=\left\{
\ba{lcl}
\gam_3,&if&\prod_{1\le i<j\le 3} (\re l_i(k)-\re l_j(k))(\re z_i(k)-\re z_j(k))<0,\\
\gam_4,&if&\prod_{1\le i<j\le 3} (\re l_i(k)-\re l_j(k))(\re z_i(k)-\re z_j(k))>0.
\ea
\right.
\ee
The rule chosen in the produce is if $l_m=l_n$, $m$ may not equals $n$, we just choose the subscript is smaller one.
\par
According to the definition of the $\gam^n$, we find that
\be\label{gamndef}
\ba{ll}
\gam^1=\left(\ba{lll}\gam_4&\gam_4&\gam_4\\\gam_2&\gam_4&\gam_4\\\gam_2&\gam_4&\gam_4\ea\right)&
\gam^2=\left(\ba{lll}\gam_3&\gam_3&\gam_3\\\gam_1&\gam_3&\gam_3\\\gam_1&\gam_3&\gam_3\ea\right)\\
\gam^3=\left(\ba{lll}\gam_3&\gam_1&\gam_1\\\gam_3&\gam_3&\gam_3\\\gam_3&\gam_3&\gam_3\ea\right)&
\gam^4=\left(\ba{lll}\gam_4&\gam_2&\gam_2\\\gam_4&\gam_4&\gam_4\\\gam_4&\gam_4&\gam_4\ea\right).
\ea
\ee
\par
The following proposition ascertains that the $M_n$'s defined in this way have the properties required for the formulation of a Riemann-Hilbert problem.
\begin{proposition}
For each $n=1,\ldots,4$, the function $M_n(x,t,k)$ is well-defined by equation (\ref{Mndef}) for $k\in \bar D_n$ and $(x,t)\in \Om$. For any fixed point $(x,t)$, $M_n$ is bounded and analytic as a function of $k\in D_n$ away from a possible discrete set of singularities $\{k_j\}$ at which the Fredholm determinant vanishes. Moreover, $M_n$ admits a bounded and contious extension to $\bar D_n$ and
\be\label{Mnasy}
M_n(x,t,k)=\id+O(\frac{1}{k}),\qquad k\rightarrow \infty,\quad k\in D_n.
\ee
\end{proposition}
\begin{proof}
The bounedness and analyticity properties are established in appendix B in \cite{l1}. And substituting the expansion
\[
M=M_0+\frac{M^{(1)}}{k}+\frac{M^{(2)}}{k^2}+\cdots,\qquad k\rightarrow \infty.
\]
into the Lax pair (\ref{muLaxe}) and comparing the terms of the same order of $k$ yield the equation (\ref{Mnasy}).
\end{proof}

\subsection{The jump matrices}

We define matrix-value functions $S_n(k)$, $n=1,\ldots,4$, and
\be\label{Sndef}
S_n(k)=M_n(0,0,k),\qquad k\in D_n,\quad n=1,\ldots,4.
\ee
Let $M$ denote the sectionally analytic function on the complex $k-$sphere which equals $M_n$ for $k\in D_n$. Then $M$ satisfies the jump conditions
\be\label{Mjump}
M_n=M_mJ_{m,n},\qquad k\in \bar D_n\cap \bar D_m,\qquad n,m=1,\ldots,4,\quad n\ne m,
\ee
where the jump matrices $J_{m,n}(x,t,k)$ are defined by
\be\label{Jmndef}
J_{m,n}=e^{(i kx+2i k^2t)\hat \Lam}(S_m^{-1}S_n).
\ee

\subsection{The adjugated eigenfunctions}
We will also need the analyticity and boundedness properties of the minors of the matrices $\{\mu_j(x,t,k)\}_1^4$. We recall that the cofactor matrix $X^A$ of a $3\times 3$ matrix $X$ is defined by
\[
X^A=\left(
\ba{ccc}
m_{11}(X)&-m_{12}(X)&m_{13}(X)\\
-m_{21}(X)&m_{22}(X)&-m_{23}(X)\\
m_{31}(X)&-m_{32}(X)&m_{33}(X)
\ea
\right),
\]
where $m_{ij}(X)$ denote the $(ij)$th minor of $X$.
\par
It follows from (\ref{muLaxe}) that the adjugated eigenfunction $\mu^A$ satisfies the Lax pair
\be\label{muadgLaxe}
\left\{
\ba{l}
\mu_x^A+[ik\Lam,\mu^A]=-V_1^T\mu^A,\\
\mu_t^A+[2ik^2\Lam,\mu^A]=-V_2^T\mu^A.
\ea
\right.
\ee
where $V^T$ denote the transform of a matrix $V$.
Thus, the eigenfunctions $\{\mu_j^A\}_1^4$ are solutions of the integral equations
\be\label{muadgdef}
\mu_j^A(x,t,k)=\id-\int_{\gam_j}e^{-ik(x-x')-2ik^2(t-t')\hat \Lam}(V_1^Tdx+V_2^T)\mu^A,\quad j=1,2,3,4.
\ee
Then we can get the following analyticity and boundedness properties:
\be\label{mujadgbodanydom}
\ba{ll}
\mu_1^A:&(D_3,D_2,D_2),\\
\mu_2^A:&(D_4,D_1,D_1),\\
\mu_3^A:&(D_2,D_3,D_3),\\
\mu_3^A:&(D_1,D_4,D_4).
\ea
\ee

\subsection{Symmetries}
We will show that the eigenfunctions $\mu_j(x,t,k)$ satisfy an important symmetry.
\begin{lemma}\label{symmetry}
The eigenfunction $\Psi(x,t,k)$ of the Lax pair (\ref{Laxpair}) satisfies the following symmetry:
\be\label{symm}
\Psi^{-1}(x,t,k)=A\ol{\Psi(x,t,\bar k)}^{T}A,
\ee
where
\be
A=\left(
\ba{ccc}
1&0&0\\
0&-\sig&0\\
0&0&-\sig
\ea
\right),\quad \sig^2=1.
\ee
Here, the superscript $T$ denotes a matrix transpose.
\end{lemma}
\begin{proof}
The equation (\ref{symm}) follows from the fact
\be
-A\ol{U(x,t,\bar k)}A=U(x,t,k)^{T},\quad -A\ol{V(x,t,\bar k)}A=V(x,t,k)^{T},
\ee
and
\be
\Psi_x^A(x,t,k)=-U(x,t,k)^T\Psi^A(x,t,k),\quad \Psi_t^A(x,t,k)=-V(x,t,k)^T\Psi^A(x,t,k)
\ee
\end{proof}

\begin{remark}
This lemma implies that the eigenfunctions $\mu_j(x,t,k)$ of Lax pair (\ref{muLaxe}) satisfy the same symmetry.
\end{remark}
\subsection{The $J_{m,n}$'s computation}

Let us define the $3\times 3-$matrix value spectral functions $s(k)$, $S(k)$ and $S_L(k)$ by
\begin{subequations}\label{sSdef}
\be\label{mu3mu2s}
\mu_3(x,t,k)=\mu_2(x,t,k)e^{(ikx+2ik^2t)\hat \Lam}s(k),
\ee
\be\label{mu1mu2S}
\mu_1(x,t,k)=\mu_2(x,t,k)e^{(ikx+2ik^2t)\hat \Lam}S(k),
\ee
\be\label{mu4mu3SL}
\mu_4(x,t,k)=\mu_3(x,t,k)e^{(ik(x-L)+2ik^2t)\hat \Lam}S_L(k)
\ee
\end{subequations}
Thus,
\begin{subequations}
\be\label{smu3}
s(k)=\mu_3(0,0,k),
\ee
\be\label{Smu1}
S(k)=\mu_1(0,0,k)=e^{-2ik^2 T\hat \Lam}\mu_2^{-1}(0,T,k),
\ee
\be\label{SLmu4}
S_L(k)=\mu_4(L,0,k)=e^{-2ik^2 T\hat \Lam}\mu_3^{-1}(L,T,k).
\ee
\end{subequations}
And we deduce from the properties of $\mu_j$ and $\mu_j^A$ that $\{s(k),S(k),S_L(k)\}$ and $\{s^A(k),S^A(k),S^A_L(k)\}$ have the following boundedness properties:
\[
\ba{ll}
s(k):&(D_3\cup D_4,D_1\cup D_2,D_1\cup D_2),\\
S(k):&(D_2\cup D_4,D_1\cup D_3,D_1\cup D_3),\\
S_L(k):&(D_2\cup D_4,D_1\cup D_3,D_1\cup D_3)\\
s^A(k):&(D_1\cup D_2,D_3\cup D_4,D_3\cup D_4),\\
S^A(k):&(D_1\cup D_3,D_2\cup D_4,D_2\cup D_4),\\
S_L^A(k):&(D_1\cup D_3,D_2\cup D_4,D_2\cup D_4).
\ea
\]
Moreover, noticing that
\be\label{MnSnrel}
M_n(x,t,k)=\mu_2(x,t,k)e^{(ikx+2ik^2t)\hat\Lam}S_n(k),\quad k\in D_n.
\ee
\begin{proposition}
The $S_n$ can be expressed in terms of the entries of $s(k),S(k)$ and $S_L(k)$ as follows:
\begin{subequations}\label{Sn}
\be
\ba{ll}
S_1=\left(\ba{ccc}\frac{1}{m_{11}(\mathcal{A})}&\mathcal{A}_{12}&\mathcal{A}_{13}\\0&\mathcal{A}_{22}
&\mathcal{A}_{23}\\0&\mathcal{A}_{32}&\mathcal{A}_{33}\ea\right),&
S_2=\left(\ba{ccc}\frac{S_{11}}{(S^Ts^A)_{11}}&s_{12}&s_{13}\\\frac{S_{21}}{(S^Ts^A)_{11}}&s_{22}
&s_{23}\\\frac{S_{31}}{(S^Ts^A)_{11}}&s_{32}&s_{33}\ea\right),\\
\ea
\ee
\be
\ba{l}
S_3=\left(\ba{ccc}s_{11}&\frac{m_{33}(s)m_{21}(S)-m_{23}(s)m_{31}(S)}{(s^TS^A)_{11}}&\frac{m_{32}(s)m_{21}(S)-m_{22}(s)m_{31}(S)}{(s^TS^A)_{11}}\\
s_{21}&\frac{m_{33}(s)m_{11}(S)-m_{13}(s)m_{31}(S)}{(s^TS^A)_{11}}&\frac{m_{32}(s)m_{11}(S)-m_{12}(s)m_{31}(S)}{(s^TS^A)_{11}}\\
s_{31}&\frac{m_{23}(s)m_{11}(S)-m_{13}(s)m_{21}(S)}{(s^TS^A)_{11}}&\frac{m_{22}(s)m_{11}(S)-m_{12}(s)m_{21}(S)}{(s^TS^A)_{11}}\ea\right),\\
S_4=\left(\ba{ccc}\mathcal{A}_{11}&0&0\\\mathcal{A}_{21}&\frac{m_{33}(\mathcal{A})}{\mathcal{A}_{11}}&\frac{m_{32}(\mathcal{A})}{\mathcal{A}_{11}}\\\mathcal{A}_{31}&\frac{m_{23}(\mathcal{A})}{\mathcal{A}_{11}}&\frac{m_{22}(\mathcal{A})}{\mathcal{A}_{11}}\ea\right).
\ea
\ee
\end{subequations}
where $\mathcal{A}=(\mathcal{A}_{ij})_{i,j=1}^{3}$ is a $3\times 3$ matrix, which is defined as
$
\mathcal{A}=s(k)e^{-ikL\hat \Lam}S_L(k).
$
And the functions
\[
(S^Ts^A)_{11}=S_{11}m_{11}(s)-S_{21}m_{21}(s)+S_{31}m_{31}(s),
\]
\[
(s^TS^A)_{11}=s_{11}m_{11}(S)-s_{21}m_{21}(S)+s_{31}m_{31}(S).
\]
\end{proposition}
\begin{proof}
Firstly, we define $R_n(k),T_n(k)$ and $Q_n(k)$ as follows:
\begin{subequations}\label{RnTnQn}
\be\label{Rn}
R_n(k)=e^{-2ik^2T\hat \Lam}M_n(0,T,k),
\ee
\be\label{Tn}
T_n(k)=e^{-ikL\hat \Lam}M_n(L,0,k),
\ee
\be\label{Qn}
Q_n(k)=e^{-(ikL+2ik^2T)\hat\Lam}M_n(L,T,k).
\ee
\end{subequations}
Then, we have the following relations:
\be\label{MnRnSnTn}
\left\{
\ba{l}
M_n(x,t,k)=\mu_1(x,t,k)e^{(ikx+2ik^2t)\hat \Lam}R_n(k),\\
M_n(x,t,k)=\mu_2(x,t,k)e^{(ikx+2ik^2t)\hat \Lam}S_n(k),\\
M_n(x,t,k)=\mu_3(x,t,k)e^{(ikx+2ik^2t)\hat \Lam}T_n(k),\\
M_n(x,t,k)=\mu_3(x,t,k)e^{(ikx+2ik^2t)\hat \Lam}Q_n(k)
\ea
\right.
\ee

The relations (\ref{MnRnSnTn}) imply that
\be\label{sSRnSnTn}
\ba{l}
s(k)=S_n(k)T^{-1}_n(k),\\
S(k)=S_n(k)R^{-1}_n(k),\\
\mathcal{A}(k)=S_n(k)Q^{-1}_n(k).
\ea
\ee
These equations constitute a matrix factorization problem which, given $\{s(k),S(k),S_L(k)\}$ can be solved for the $\{R_n,S_n,T_n,Q_n\}$. Indeed, the integral equations (\ref{Mndef}) together with the definitions of $\{R_n,S_n,T_n,Q_n\}$ imply that
\be
\left\{
\ba{lll}
(R_n(k))_{ij}=0&if&\gam_{ij}^n=\gam_1,\\
(S_n(k))_{ij}=0&if&\gam_{ij}^n=\gam_2,\\
(T_n(k))_{ij}=\dta_{ij}&if&\gam_{ij}^n=\gam_3,\\
(Q_n(k))_{ij}=\dta_{ij}&if&\gam_{ij}^n=\gam_4.
\ea
\right.
\ee
It follows that (\ref{sSRnSnTn}) are 27 scalar equations for 27 unknowns. By computing the explicit solution of this algebraic system, we arrive at (\ref{Sn}).
\end{proof}

\begin{remark}
Due to our symmetry, see Lemma \ref{symmetry}, obtained in the above subsection we can replace the minors by conjugate terms among the representation of the functions $S_n(k)$. It may looks like much simple to compute the jump matrices $J_{m,n}(x,t,k)$.
\end{remark}

\subsection{The residue conditions}
Since $\mu_2$ is an entire function, it follows from (\ref{MnSnrel}) that M can only have sigularities at the points where the $S_n's$ have singularities.
We denote the possible zeros by $\{k_j\}_1^N$ and assume they satisfy the following assumption.
\begin{assumption}\label{assum}
We assume that
\begin{itemize}
\item $m_{11}(\mathcal{A})(k)$ has $n_0$ possible simple zeros in $D_1$ denoted by $\{k_j\}_1^{n_0}$;
\item $(S^Ts^A)_{11}(k)$ has $n_1-n_0$ possible simple zeros in $D_2$ denoted by $\{k_j\}_{n_0+1}^{n_1}$;
\item $(s^TS^A)_{11}(k)$ has $n_2-n_1$ possible simple zeros in $D_3$ denoted by $\{k_j\}_{n_1+1}^{n_2}$;
\item $\mathcal{A}_{11}(k)$ has $N-n_2$ possible simple zeros in $D_4$ denoted by $\{k_j\}_{n_2+1}^{N}$;
\end{itemize}
and that none of these zeros coincide. Moreover, we assume that none of these functions have zeros on the boundaries of the $D_n$'s.
\end{assumption}
We determine the residue conditions at these zeros in the following:
\begin{proposition}\label{propos}
Let $\{M_n\}_1^4$ be the eigenfunctions defined by (\ref{Mndef}) and assume that the set $\{k_j\}_1^N$ of singularities are as the above assumption. Then the following residue conditions hold:
\begin{subequations}
\be\label{M11D1res}
{Res}_{k=k_j}[M]_1=\frac{\mathcal{A}_{33}(k_j)[M(k_j)]_2-\mathcal{A}_{23}(k_j)[M(k_j)]_3}{\dot m_{11}(\mathcal{A})(k_j)m_{21}(\mathcal{A})(k_j)}e^{2\tha(k_j)},\quad 1\le j\le n_0,k_j\in D_1
\ee
\be\label{M21D2res}
\ba{r}
Res_{k=k_j}[M]_1=\frac{S_{21}(k_j)s_{33}(k_j)-S_{31}(k_j)s_{23}(k_j)}{\dot{(S^Ts^A)_{33}(k_j)}m_{11}(k_j)}e^{2\tha(k_j)}[M(k_j)]_2\\
{}+\frac{S_{31}(k_j)s_{22}(k_j)-S_{21}(k_j)s_{32}(k_j)}{\dot{(S^Ts^A)_{33}(k_j)}m_{11}(k_j)}e^{2\tha(k_j)}[M(k_j)]_3\\
\quad n_0+1\le j\le n_1,k_j\in D_2,
\ea
\ee
\be\label{M32D3res}
\ba{r}
Res_{k=k_j}[M]_2=\frac{m_{33}(s)(k_j)M_{21}(S)(k_j)-m_{23}(s)(k_j)M_{31}(S)(k_j)}{\dot{(s^TS^A)_{11}(k_j)}s_{11}(k_j)}e^{-2\tha(k_j)}[M(k_j)]_1\\
\quad n_1+1\le j\le n_2,k_j\in D_3,
\ea
\ee

\be\label{M33D3res}
\ba{r}
Res_{k=k_j}[M]_3=\frac{m_{32}(s)(k_j)M_{21}(S)(k_j)-m_{22}(s)(k_j)M_{31}(S)(k_j)}{\dot{(s^TS^A)_{11}(k_j)}s_{11}(k_j)}e^{-2\tha(k_j)}[M(k_j)]_1\\
\quad n_1+1\le j\le n_2,k_j\in D_3.
\ea
\ee
\be\label{M42D4res}
Res_{k=k_j}[M]_2=\frac{m_{33}(s)(k_j)}{\dot s_{11}(k_j) s_{21}(k_j)}e^{-2\tha(k_j)}[M(k_j)]_1,
\quad n_2+1\le j\le N,k_j\in D_4.
\ee
\be\label{M43D4res}
Res_{k=k_j}[M]_3=\frac{m_{32}(s)(k_j)}{\dot s_{11}(k_j) s_{21}(k_j)}e^{-2\tha(k_j)}[M(k_j)]_1,
\quad n_2+1\le j\le N,k_j\in D_4.
\ee
\end{subequations}
where $\dot f=\frac{df}{dk}$, and $\tha$ is defined by
\be\label{thaijdef}
\tha(x,t,k)=ikx+2ik^2t.
\ee
\end{proposition}
\begin{proof}
We will prove (\ref{M11D1res}),  (\ref{M32D3res}),  the other conditions follow by similar arguments.
Equation (\ref{MnSnrel}) implies the relation
\begin{subequations}
\be\label{M1S1}
M_1=\mu_2e^{(ikx+2ik^2t)\hat\Lam}S_1,
\ee

\be\label{M3S3}
M_3=\mu_2e^{(ikx+2ik^2t)\hat\Lam}S_3,
\ee

\end{subequations}
In view of the expressions for $S_1$ and $S_3$ given in (\ref{Sn}), the three columns of (\ref{M1S1}) read:
\begin{subequations}
\be\label{M11}
[M_1]_1=[\mu_2]_1\frac{1}{m_{11}(\mathcal{A})},
\ee
\be\label{M12}
[M_1]_2=[\mu_2]_1e^{-2\tha}\mathcal{A}_{12}+[\mu_2]_2\mathcal{A}_{22}+[\mu_2]_3\mathcal{A}_{32},
\ee
\be\label{M13}
[M_1]_3=[\mu_2]_1e^{-2\tha}\mathcal{A}_{13}+[\mu_2]_2\mathcal{A}_{23}+[\mu_2]_3\mathcal{A}_{33}.
\ee
\end{subequations}
while the three columns of (\ref{M3S3}) read:
\begin{subequations}
\be\label{M31}
[M_3]_1=[\mu_2]_1s_{11}+[\mu_2]_2s_{21}e^{2\tha}+[\mu_2]_3s_{31}e^{2\tha}
\ee
\be\label{M32}
\ba{rl}
[M_3]_2&=[\mu_2]_1\frac{m_{33}(s)m_{21}(S)-m_{23}(s)m_{31}(S)}{(s^TS^A)_{11}}e^{-2\tha}\\
&+[\mu_2]_2\frac{m_{33}(s)m_{11}(S)-m_{13}(s)m_{31}(S)}{(s^TS^A)_{11}}\\
&+[\mu_2]_3\frac{m_{23}(s)m_{11}(S)-m_{13}(s)m_{21}(S)}{(s^TS^A)_{11}}
\ea
\ee
\be\label{M33}
\ba{rl}
[M_3]_3&=[\mu_2]_1\frac{m_{32}(s)m_{21}(S)-m_{22}(s)m_{31}(S)}{(s^TS^A)_{11}}e^{-2\tha}\\
&+[\mu_2]_2\frac{m_{32}(s)m_{11}(S)-m_{12}(s)m_{31}(S)}{(s^TS^A)_{11}}\\
&+[\mu_2]_3\frac{m_{22}(s)m_{11}(S)-m_{12}(s)m_{21}(S)}{(s^TS^A)_{11}}.
\ea
\ee
\end{subequations}

We first suppose that $k_j\in D_1$ is a simple zero of $m_{11}(\mathcal{A})(k)$. Solving (\ref{M12}) and (\ref{M13}) for $[\mu_2]_1,[\mu_2]_3$ and substituting the result in to (\ref{M11}), we find
\[
[M_1]_1=\frac{\mathcal{A}_{33}[M_1]_2-\mathcal{A}_{32}[M_1]_3}{m_{11}(\mathcal{A})m_{21}(\mathcal{A})}e^{2\tha}-\frac{[\mu_2]_2}{m_{21}(\mathcal{A})}e^{2\tha}.
\]
Taking the residue of this equation at $k_j$, we find the condition (\ref{M11D1res}) in the case when $k_j\in D_1$.
\par
In order to prove (\ref{M32D3res}), we solve (\ref{M31}) for $[\mu_2]_1$, then substituting the result into (\ref{M32}) and (\ref{M33}), we find
\begin{subequations}
\be
[M_3]_2=\frac{m_{33}(s)}{s_{11}}[\mu_2]_2+\frac{m_{23}(s)}{s_{11}}[\mu_2]_3+\frac{m_{33}(s)m_{21}(S)-m_{23}(s)m_{31}(S)}{\dot {(s^TS^A)_{11}}s_{11}}e^{-2\tha}[M_3]_1,
\ee
\be
[M_3]_3=\frac{m_{32}(s)}{s_{11}}[\mu_2]_2+\frac{m_{22}(s)}{s_{11}}[\mu_2]_3+\frac{m_{32}(s)m_{21}(S)-m_{22}(s)m_{31}(S)}{\dot {(s^TS^A)_{11}}s_{11}}e^{-2\tha}[M_3]_1.
\ee
\end{subequations}
Taking the residue of this equation at $k_j$, we find the condition (\ref{M32D3res}) in the case when $k_j\in D_3$. 
\end{proof}

\subsection{The global relation}
The spectral functions $S(k),S_L(k)$ and $s(k)$ are not independent but satisfy an important relation. Indeed, it follows from (\ref{sSdef}) that
\be
\mu_1(x,t,k)e^{(ikx+2ik^2t)\hat \Lam}\{S^{-1}(k)s(k)e^{-ikL\hat\Lam}S_L(k)\}=\mu_4(x,t,k).
\ee
Since $\mu_1(0,T,k)=\id$, evaluation at $(0,T)$ yields the following global relation:
\be\label{globalrel}
S^{-1}(k)s(k)e^{-ikL\hat\Lam}S_L(k)=e^{-2ik^2T\hat \Lam}c(T,k),
\ee
where $c(T,k)=\mu_4(0,T,k)$.

\section{The Riemann-Hilbert problem}

The sectionally analytic function $M(x,t,k)$ defined in section 2 satisfies a Riemann-Hilbert problem which can be formulated in terms of the initial and boundary values of $q_1(x,t)$ and $q_2(x,t)$. By solving this Riemann-Hilbert problem, the solution of (\ref{2-NLSe}) can be recovered for all values of $x,t$.
\begin{theorem}
Suppose that $q_1(x,t)$ and $q_2(x,t)$ are a pair of solutions of (\ref{2-NLSe}) in the interval domain $\Om$. Then $q_1(x,t)$ and $q_2(x,t)$ can be reconstructed from the initial value $\{q_{10}(x),q_{20}(x)\}$ and boundary values $\{g_{01}(t),g_{02}(t),g_{11}(t),g_{12}(t)\}$, $\{f_{01}(t),f_{02}(t),f_{11}(t),f_{12}(t)\}$ defined as follows,
\be\label{inibouvalu}
\ba{ll}
q_{10}(x)=q_1(x,t=0),& q_{20}(x)=q_2(x,t=0),\\
g_{01}(t)=q_1(x=0,t),& g_{02}(t)=q_2(x=0,t),\\
f_{01}(t)=q_1(x=L,t),& f_{02}(t)=q_2(x=L,t),\\
g_{11}(t)=q_{1x}(x=0,t),& g_{12}(t)=q_{2x}(x=0,t),\\
f_{11}(t)=q_{1x}(x=L,t),& f_{12}(t)=q_{2x}(x=L,t).
\ea
\ee
\par
Use the initial and boundary data to define the jump matrices $J_{m,n}(x,t,k)$ in terms of the spectral functions $s(k)$ and $S(k),S_L(k)$ by equation (\ref{sSdef}).
\par
Assume that the possible zeros $\{k_j\}_1^N$ of the functions $m_{11}(\mathcal{A})(k)$, $(S^Ts^A)_{11}(k)$, $(s^TS^A)_{11}(k)$ and $\mathcal{A}_{11}(k)$ are as in assumption \ref{assum}.
\par
Then the solution $\{q_1(x,t),q_2(x,t)\}$ is given by
\be\label{usolRHP}
q_1(x,t)=2i\lim_{k\rightarrow \infty}(kM(x,t,k))_{12},\quad q_2(x,t)=2i\lim_{k\rightarrow \infty}(kM(x,t,k))_{13} .
\ee
where $M(x,t,k)$ satisfies the following $3\times 3$ matrix Riemann-Hilbert problem:
\begin{itemize}
\item $M$ is sectionally meromorphic on the Riemann $k-$sphere with jumps across the contours $\bar D_n\cap \bar D_m,n,m=1,\cdots, 4$, see Figure \ref{fig2}.
\item Across the contours $\bar D_n\cap \bar D_m$, $M$ satisfies the jump condition
      \be\label{MRHP}
      M_n(x,t,k)=M_m(x,t,k)J_{m,n}(x,t,k),\quad k\in \bar D_n\cap \bar D_m,n,m=1,2,3,4.
      \ee
\item $M(x,t,k)=\id+O(\frac{1}{k}),\qquad k\rightarrow \infty$.
\item The residue condition of $M$ is showed in Proposition \ref{propos}.
\end{itemize}
\end{theorem}
\begin{proof}
It only remains to prove (\ref{usolRHP}) and this equation follows from the large $k$ asymptotics of the eigenfunctions.
\end{proof}

\section{Non-linearizable Boundary Conditions}
A major difficulty of initial-boundary value problems is that some of the boundary values are unkown for a well-posed problem. All boundary values are needed for the definition of $S(k), S_L(k)$, and hence for the formulation of the Riemann-Hilbert problem. Our main result expresses the unknown boundary data in terms of the prescribed boundary data and the initial data via the solution of a system of nonlinear integral equations.

\subsection{Asymptotics}

An analysis of (\ref{muLaxe}) shows that the eigenfunctions $\{\mu_j\}_1^4$ have the following asymptotics as $k\rightarrow\infty$:

\be\label{mujasykinf}
\ba{l}
\mu_j(x,t,k)=\id+
\frac{1}{k}\left(\ba{lll}\mu^{(1)}_{11}&\mu^{(1)}_{12}&\mu^{(1)}_{13}\\\mu^{(1)}_{21}&\mu^{(1)}_{22}&\mu^{(1)}_{23}\\\mu^{(1)}_{31}&\mu^{(1)}_{32}&\mu^{(1)}_{33}\ea\right)
+\frac{1}{k^2}\left(\ba{lll}\mu^{(2)}_{11}&\mu^{(2)}_{12}&\mu^{(2)}_{13}\\\mu^{(2)}_{21}&\mu^{(2)}_{22}&\mu^{(2)}_{23}\\\mu^{(2)}_{31}&\mu^{(2)}_{32}&\mu^{(2)}_{33}\ea\right)
+O(\frac{1}{k^3})\\
=\id+\frac{1}{k}\left(\ba{ccc}\int_{(x_j,t_j)}^{(x,t)} \Dta_{11}&\frac{q_1}{2i}&\frac{q_2}{2i}\\
-\frac{\sig\bar q_1}{2i}&\int_{(x_j,t_j)}^{(x,t)}\Dta^{(1)}_{22}&\int_{(x_j,t_j)}^{(x,t)}\Dta^{(1)}_{23}\\
-\frac{\sig\bar q_2}{2i}&\int_{(x_j,t_j)}^{(x,t)}\Dta^{(1)}_{32}&\int_{(x_j,t_j)}^{(x,t)}\Dta^{(1)}_{33} \ea\right)\\
+\frac{1}{k^2}\left(\ba{ccc}
\mu^{(2)}_{11}&\frac{q_1\mu^{(1)}_{22}+q_2\mu^{(1)}_{32}}{2i}+\frac{1}{4}q_{1x}&\frac{q_1\mu^{(1)}_{23}+q_2\mu^{(1)}_{33}}{2i}+\frac{1}{4}q_{2x}\\
\frac{1}{4}\sig \bar q_{1x}-\frac{\sig}{2i}\bar q_{1}\mu^{(1)}_{11}&\mu^{(2)}_{22}&\mu^{(2)}_{23}\\
\frac{1}{4}\sig \bar q_{2x}-\frac{\sig}{2i}\bar q_{2}\mu^{(1)}_{11}&\mu^{(2)}_{32}&\mu^{(2)}_{33}
\ea\right)+O(\frac{1}{k^3}).
\ea
\ee

where

\be\label{Dtadef}
\ba{l}
\Dta_{11}=\sig[\frac{i}{2}(|q_1|^2+|q_2|^2)dx+\frac{1}{2}(q_1\bar q_{1x}-\bar q_1 q_{1x}+q_2 \bar q_{2x}-\bar q_2 q_{2x})dt]\\
\Dta^{(1)}_{22}=\sig[-\frac{i}{2}|q_1|^2dx-\frac{1}{2}(\bar q_{1x}q_1-\bar q_1 q_{1x})dt]\\
\Dta^{(1)}_{23}=\sig[-\frac{i}{2}\bar q_1 q_2dx-\frac{1}{2}(\bar q_{1x}q_2-\bar q_1 q_{2x})dt]\\
\Dta^{(1)}_{32}=\sig[-\frac{i}{2}\bar q_2 q_1dx-\frac{1}{2}(\bar q_{2x}q_1-\bar q_2 q_{1x})dt]\\
\Dta^{(1)}_{33}=\sig[-\frac{i}{2}|q_2|^2dx-\frac{1}{2}(\bar q_{2x}q_2-\bar q_2 q_{2x})dt].
\ea
\ee

The functions $\{\mu^{(i)}_{jl}=\mu^{(i)}_{jl}(x,t)\}_{j,l=1}^{3},i=1,2$ are independent of $k$.
\begin{remark}
Because we do not need the asymptotic expressions of the $\mu^{(2)}_{11}$ and $\{\mu^{(2)}_{ij}\}_{i,j=2}^{3}$ in the following analysis, we do not write down the explicit formulas for these functions.
\end{remark}
We define functions $\{\Phi_{ij}(t,k)\}_{i,j=1}^{3}$ and $\{\phi_{ij}(t,k)\}_{i,j=1}^{3}$ by:
\be
\mu_2(0,t,k)=\left(\ba{lll}\Phi_{11}(t,k)&\Phi_{12}(t,k)&\Phi_{13}(t,k)\\
\Phi_{21}(t,k)&\Phi_{22}(t,k)&\Phi_{23}(t,k)\\\Phi_{31}(t,k)&\Phi_{32}(t,k)&\Phi_{33}(t,k)\ea\right),
\ee
\be
\mu_3(L,t,k)=\left(\ba{lll}\phi_{11}(t,k)&\phi_{12}(t,k)&\phi_{13}(t,k)\\
\phi_{21}(t,k)&\phi_{22}(t,k)&\phi_{23}(t,k)\\\phi_{31}(t,k)&\phi_{32}(t,k)&\phi_{33}(t,k)\ea\right).
\ee

\par
From the asymptotic of $\mu_j(x,t,k)$ in (\ref{mujasykinf}) we have
\be\label{mu2x0tk}
\ba{rcl}
\mu_{2}(0,t,k)&=&\id+
\frac{1}{k}\left(\ba{ccc}\Phi^{(1)}_{11}(t)&\Phi^{(1)}_{12}(t)&\Phi^{(1)}_{13}(t)\\
\Phi^{(1)}_{21}(t)&\Phi^{(1)}_{22}(t)&\Phi^{(1)}_{23}(t)\\
\Phi^{(1)}_{31}(t)&\Phi^{(1)}_{32}(t)&\Phi^{(1)}_{33}(t)
\ea\right)\\
&&{}+\frac{1}{k^2}\left(\ba{ccc}\Phi^{(2)}_{11}(t)&\Phi^{(2)}_{12}(t)&\Phi^{(2)}_{13}(t)\\
\Phi^{(2)}_{21}(t)&\Phi^{(2)}_{22}(t)&\Phi^{(2)}_{23}(t)\\
\Phi^{(2)}_{31}(t)&\Phi^{(2)}_{32}(t)&\Phi^{(2)}_{33}(t)
\ea
\right)+O(\frac{1}{k^3}).
\ea
\ee
Recalling that the definition of the boundary data at $x=0$, we have
\be
\ba{ll}
\Phi_{12}^{(1)}(t)=\frac{1}{2i}g_{01}(t),&\Phi_{12}^{(2)}(t)=\frac{1}{4}g_{11}(t)+\frac{(g_{01}(t)\Phi^{(1)}_{22}(t)+g_{02}(t)\Phi^{(1)}_{32}(t))}{2i},\\
\Phi_{13}^{(1)}(t)=\frac{1}{2i}g_{02}(t),&\Phi_{13}^{(2)}(t)=\frac{1}{4}g_{12}(t)+\frac{(g_{01}(t)\Phi^{(1)}_{23}(t)+g_{02}(t)\Phi^{(1)}_{33}(t))}{2i},\\
\Phi_{22}^{(1)}(t)=-\frac{\sig}{2}\int_{0}^{t}(\bar g_{11}(t)g_{01}(t)-\bar g_{01}(t)g_{11}(t))dt,&\Phi_{23}^{(1)}(t)=-\frac{\sig}{2}\int_{0}^{t}(\bar g_{11}(t)g_{02}(t)-\bar g_{01}(t)g_{12}(t))dt,\\
\Phi_{32}^{(1)}(t)=-\frac{\sig}{2}\int_{0}^{t}(\bar g_{12}(t)g_{01}(t)-\bar g_{02}(t)g_{11}(t))dt,&\Phi_{33}^{(1)}(t)=-\frac{\sig}{2}\int_{0}^{t}(\bar g_{12}(t)g_{02}(t)-\bar g_{02}(t)g_{12}(t))dt.
\ea
\ee

\par
In particular, we find the following expressions for the boudary values at $x=0$:
\begin{subequations}\label{g01}
\be\label{g0}
g_{01}(t)=2i\Phi_{12}^{(1)}(t),\quad
g_{02}(t)=2i\Phi_{13}^{(1)}(t),
\ee
\be\label{g1}
\ba{l}
g_{11}(t)=4\Phi_{12}^{(2)}(t)+2i(g_{01}(t)\Phi_{22}^{(1)}(t)+g_{02}(t)\Phi_{32}^{(1)}(t)),\\
g_{12}(t)=4\Phi_{13}^{(2)}(t)+2i(g_{01}(t)\Phi_{23}^{(1)}(t)+g_{02}(t)\Phi_{33}^{(1)}(t))
\ea
\ee
\end{subequations}

Similarly, we have the asymptotic formulas for $\mu_3(L,t,k)=\{\phi_{ij}(t,k)\}_{i,j=1}^{3}$,
\be\label{mu3xLtk}
\ba{rcl}
\mu_{3}(L,t,k)&=&\id+
\frac{1}{k}\left(\ba{ccc}\phi^{(1)}_{11}(t)&\phi^{(1)}_{12}(t)&\phi^{(1)}_{13}(t)\\
\phi^{(1)}_{21}(t)&\phi^{(1)}_{22}(t)&\phi^{(1)}_{23}(t)\\
\phi^{(1)}_{31}(t)&\phi^{(1)}_{32}(t)&\phi^{(1)}_{33}(t)
\ea\right)\\
&&{}+\frac{1}{k^2}\left(\ba{ccc}\phi^{(2)}_{11}(t)&\phi^{(2)}_{12}(t)&\phi^{(2)}_{13}(t)\\
\phi^{(2)}_{21}(t)&\phi^{(2)}_{22}(t)&\phi^{(2)}_{23}(t)\\
\phi^{(2)}_{31}(t)&\phi^{(2)}_{32}(t)&\phi^{(2)}_{33}(t)
\ea
\right)+O(\frac{1}{k^3}).
\ea
\ee
Recalling that the definition of the boundary data at $x=L$, we have
\be
\ba{ll}
\phi_{12}^{(1)}(t)=\frac{1}{2i}f_{01}(t),&\phi_{12}^{(2)}(t)=\frac{1}{4}f_{11}(t)+\frac{(f_{01}(t)\phi^{(1)}_{22}(t)+f_{02}(t)\phi^{(1)}_{32}(t))}{2i},\\
\phi_{13}^{(1)}(t)=\frac{1}{2i}f_{02}(t),&\phi_{13}^{(2)}(t)=\frac{1}{4}f_{12}(t)+\frac{(f_{01}(t)\phi^{(1)}_{23}(t)+f_{02}(t)\phi^{(1)}_{33}(t))}{2i},\\
\phi_{22}^{(1)}(t)=-\frac{\sig}{2}\int_{0}^{t}(\bar f_{11}(t)f_{01}(t)-\bar f_{01}(t)f_{11}(t))dt,&\phi_{23}^{(1)}(t)=-\frac{\sig}{2}\int_{0}^{t}(\bar f_{11}(t)f_{02}(t)-\bar f_{01}(t)f_{12}(t))dt,\\
\phi_{32}^{(1)}(t)=-\frac{\sig}{2}\int_{0}^{t}(\bar f_{12}(t)f_{01}(t)-\bar f_{02}(t)f_{11}(t))dt,&\phi_{33}^{(1)}(t)=-\frac{\sig}{2}\int_{0}^{t}(\bar f_{12}(t)f_{02}(t)-\bar f_{02}(t)f_{12}(t))dt.
\ea
\ee
In particular, we find the following expressions for the boudary values at $x=L$:
\begin{subequations}\label{f01}
\be\label{f0}
f_{01}(t)=2i\phi_{12}^{(1)}(t),\quad
f_{02}(t)=2i\phi_{13}^{(1)}(t),
\ee
\be\label{f1}
\ba{l}
f_{11}(t)=4\phi_{12}^{(2)}(t)+2i(f_{01}(t)\phi_{22}^{(1)}(t)+f_{02}(t)\phi_{32}^{(1)}(t)),\\
f_{12}(t)=4\phi_{13}^{(2)}(t)+2i(f_{01}(t)\phi_{23}^{(1)}(t)+f_{02}(t)\phi_{33}^{(1)}(t))
\ea
\ee
\end{subequations}

\par

From the global relation (\ref{globalrel})and replacing $T$ by $t$, we find
\be\label{useglo}
\mu_2(0,t,k)e^{2ik^2t\hat \Lam}\{s(k)e^{-ikL\hat\Lam}S_L(k)\}=c(t,k).
\ee
\par
From the relation (\ref{SLmu4}) and the symmetry (\ref{symm}), we know that the spectral function $S_L(k)$ can be expressed by $\{\phi_{ij}(t,k)\}_{i,j=1}^{3}$. So if we denote the matrix-value function $c(t,k)$ as $c(t,k)=\left(c_{ij}(t,k)\right)_{i,j=1}^{3}$. 
The functions $\{c_{ij}(t,k)\}_{i=1,j=2}^{3}$ are analytic and bounded in $D_1$ away from the possible zeros of $m_{11}(\mathcal{A})(k)$ and of order $O(\frac{1+e^{2ikL}}{k})$ as $k\rightarrow \infty$.

\par
In the vanishing initial value case, the asymptotic of $c_{1j}(t,k),j=2,3$ becomes much more simple.

\begin{lemma}
We assuming that the initial value and boundary value are compatible at $x=0$ and $x=L$ (i.e. at $x=0$, $q_{10}(0)=g_{01}(0),q_{20}(0)=g_{02}(0)$; at $x=L$, $q_{10}(L)=f_{01}(0),q_{20}(L)=f_{02}(0)$). Then, in the vanishing initial value case, the global relation (\ref{useglo}) implies that the large $k$ behavior of $c_{1j}(t,k),j=2,3$ satisfies
\begin{subequations}\label{c1jlargek}
\be\label{c12largek}
\ba{rl}
c_{12}(t,k)&=\frac{\Phi^{(1)}_{12}(t)}{k}+\frac{\Phi_{12}^{(2)}(t)+\Phi_{12}^{(1)}(t)\bar \phi_{22}^{(1)}(t)+\Phi_{13}^{(1)}(t)\bar \Phi_{23}^{(1)}(t)}{k^2}+O(\frac{1}{k^3})\\
{}&-\sig\left[\frac{\bar \phi_{21}^{(1)}(t)}{k}+\frac{\bar \phi_{21}^{(2)}(t)+\Phi_{11}^{(1)}(t)\bar \phi_{21}^{(1)}(t)}{k^2}+O(\frac{1}{k^3})\right]e^{2ikL}\quad k\rightarrow \infty,
\ea
\ee
\be\label{c13largek}
\ba{rl}
c_{13}(t,k)&=\frac{\Phi^{(1)}_{13}(t)}{k}+\frac{\Phi_{13}^{(2)}(t)+\Phi_{12}^{(1)}(t)\bar \phi_{32}^{(1)}(t)+\Phi_{13}^{(1)}(t)\bar \Phi_{33}^{(1)}(t)}{k^2}+O(\frac{1}{k^3})\\
{}&-\sig\left[\frac{\bar \phi_{31}^{(1)}(t)}{k}+\frac{\bar \phi_{31}^{(2)}(t)+\Phi_{11}^{(1)}(t)\bar \phi_{31}^{(1)}(t)}{k^2}+O(\frac{1}{k^3})\right]e^{2ikL} \quad k\rightarrow \infty.
\ea
\ee
\end{subequations}
\end{lemma}

\begin{proof}

%
The global relation shows that under the assumption of vanishing initial value
\begin{subequations}
\be\label{c12def}
c_{12}(t,k)=\Phi_{12}(t,k)\bar \phi_{22}(t,\bar k)+\Phi_{13}(t,k) \bar \phi_{23}(t,\bar k)-\sig \Phi_{11}\bar \phi_{21}(t,\bar k)e^{2ikL},
\ee
\be\label{c13def}
c_{13}(t,k)=\Phi_{12}(t,k)\bar \phi_{32}(t,\bar k)+\Phi_{13}(t,k) \bar \phi_{33}(t,\bar k)-\sig \Phi_{11}\bar \phi_{31}(t,\bar k)e^{2ikL}
\ee
\end{subequations}
Recalling the equation
\be\label{Phit}
\mu_t-2ik^2[\Lam,\mu]=V_2\mu.
\ee
\begin{subequations}\label{Phiteqn}
From the first column of the equation (\ref{Phit}) we get
\be\label{Phi1t}
\left\{
\ba{l}
\Phi_{11t}=2k(g_{01}\Phi_{21}+g_{02}\Phi_{31})-i\sig (|g_{01}|^2+|g_{02}|^2)\Phi_{11}+i(g_{11}\Phi_{21}+g_{12}\Phi_{31}),\\
\Phi_{21t}-4ik^2\Phi_{21}=2\sig k \bar g_{01}\Phi_{11}+i\sig (|g_{01}|^2\Phi_{21}+\bar g_{01}g_{02}\Phi_{31})-i\sig \bar g_{11}\Phi_{11},\\
\Phi_{31t}-4ik^2\Phi_{31}=2\sig k \bar g_{02}\Phi_{11}+i\sig (\bar g_{02}g_{01}\Phi_{21}+|g_{02}|^2\Phi_{31})-i\sig \bar g_{12}\Phi_{11},
\ea
\right.
\ee

From the second column of the equation (\ref{Phit}) we get
\be\label{Phi2t}
\left\{
\ba{l}
\Phi_{12t}+4ik^2\Phi_{12}=2k(g_{01}\Phi_{22}+g_{02}\Phi_{32})-i\sig (|g_{01}|^2+|g_{02}|^2)\Phi_{12}+i(g_{11}\Phi_{22}+g_{12}\Phi_{32}),\\
\Phi_{22t}=2\sig k \bar g_{01}\Phi_{12}+i\sig (|g_{01}|^2\Phi_{22}+\bar g_{01}g_{02}\Phi_{32})-i\sig \bar g_{11}\Phi_{12},\\
\Phi_{32t}=2\sig k \bar g_{02}\Phi_{12}+i\sig (\bar g_{02}g_{01}\Phi_{22}+|g_{02}|^2\Phi_{32})-i\sig \bar g_{12}\Phi_{12},
\ea
\right.
\ee
From the third column of the equation (\ref{Phit}) we get
\be\label{Phi3t}
\left\{
\ba{l}
\Phi_{13t}+4ik^2\Phi_{13}=2k(g_{01}\Phi_{23}+g_{02}\Phi_{33})-i\sig (|g_{01}|^2+|g_{02}|^2)\Phi_{13}+i(g_{11}\Phi_{23}+g_{12}\Phi_{33}),\\
\Phi_{23t}=2\sig k \bar g_{01}\Phi_{13}+i\sig (|g_{01}|^2\Phi_{23}+\bar g_{01}g_{02}\Phi_{33})-i\sig \bar g_{11}\Phi_{13},\\
\Phi_{33t}=2\sig k \bar g_{02}\Phi_{13}+i\sig (\bar g_{02}g_{01}\Phi_{23}+|g_{02}|^2\Phi_{33})-i\sig \bar g_{12}\Phi_{13},
\ea
\right.
\ee
\end{subequations}
\par
Suppose
\be\label{Phi1albt}
\left(\ba{l}\Phi_{11}\\\Phi_{21}\\\Phi_{31}\ea\right)=\left(\alpha_0(t)+\frac{\alpha_1(t)}{k}+\frac{\alpha_2(t)}{k^2}+\cdots\right)
+\left(\beta_0(t)+\frac{\beta_1(t)}{k}+\frac{\beta_2(t)}{k^2}+\cdots\right)e^{4ik^2t},
\ee
where the coefficients $\alpha_j(t)$ and $\beta_{j}(t)$, $j=0,1,2,\cdots$, are independent of $k$ and are $3\times 1$ matrix functions.
\par
To determine these coefficients,we substitute the above equation into equation (\ref{Phi1t}) and use the initial conditions
\[
\alpha_0(0)+\beta_0(0)=(\ba{ccc}1&0&0\ea)^T,\quad \alpha_1(0)+\beta_1(0)=(\ba{ccc}0&0&0\ea)^T.
\]
Then we get
\be\label{Phi2albtrsu}
\ba{l}
\left(\ba{l}\Phi_{11}\\\Phi_{21}\\\Phi_{31}\ea\right)=
\left(\ba{l}1\\0\\0\ea\right)+\frac{1}{k}\left(\ba{l}\Phi_{11}^{(1)}\\\Phi_{21}^{(1)}\\\Phi_{31}^{(1)}\ea\right)+\frac{1}{k^2}\left(\ba{l}\Phi_{11}^{(2)}\\\Phi_{21}^{(2)}\\\Phi_{31}^{(2)}\ea\right)+O(\frac{1}{k^3})\\
{}+\left[\frac{1}{k}\left(\ba{c}0\\-\Phi^{(1)}_{21}(0)\\-\Phi^{(1)}_{31}(0)\ea\right)+O(\frac{1}{k^2})\right]e^{4ik^2t}
\ea
\ee
\par
Similarly, suppose
\be\label{Phi2albt}
\left(\ba{l}\Phi_{12}\\\Phi_{22}\\\Phi_{32}\ea\right)=\left(\alpha_0(t)+\frac{\alpha_1(t)}{k}+\frac{\alpha_2(t)}{k^2}+\cdots\right)
+\left(\beta_0(t)+\frac{\beta_1(t)}{k}+\frac{\beta_2(t)}{k^2}+\cdots\right)e^{-4ik^2t},
\ee
where the coefficients $\alpha_j(t)$ and $\beta_{j}(t)$, $j=0,1,2,\cdots$, are independent of $k$ and are $3\times 1$ matrix functions.
\par
To determine these coefficients,we substitute the above equation into equation (\ref{Phi2t}) and use the initial conditions
\[
\alpha_0(0)+\beta_0(0)=(\ba{ccc}0&1&0\ea)^T,\quad \alpha_1(0)+\beta_1(0)=(\ba{ccc}0&0&0\ea)^T.
\]
Then we get
\be\label{Phi2albtrsu}
\ba{l}
\left(\ba{l}\Phi_{12}\\\Phi_{22}\\\Phi_{32}\ea\right)=
\left(\ba{l}0\\1\\0\ea\right)+\frac{1}{k}\left(\ba{l}\Phi_{12}^{(1)}\\\Phi_{22}^{(1)}\\\Phi_{32}^{(1)}\ea\right)+\frac{1}{k^2}\left(\ba{l}\Phi_{12}^{(2)}\\\Phi_{22}^{(2)}\\\Phi_{32}^{(2)}\ea\right)+O(\frac{1}{k^3})\\
{}+\left[\frac{1}{k}\left(\ba{c}-\Phi^{(1)}_{12}(0)\\0\\0\ea\right)+
\frac{1}{k^2}\left(\ba{c}-\Phi^{(2)}_{12}(0)+\Phi^{(1)}_{12}(0)\Phi^{(1)}_{22}+\Phi^{(1)}_{12}(0)\Phi^{(1)}_{32}\\-\frac{i\sig}{2}\bar g_{01}\Phi^{(1)}_{12}(0)\\-\frac{i\sig}{2}\bar g_{02}\Phi^{(1)}_{12}(0)\ea\right)+O(\frac{1}{k^2})\right]e^{-4ik^2t}
\ea
\ee

Similar to the derivation of $\Phi_{i2},i=1,2,3$, from (\ref{Phi3t}) we can get the asymptotic formulas of $\Phi_{i3},i=1,2,3$
\be\label{Phi3albtrsu}
\ba{l}
\left(\ba{l}\Phi_{13}\\\Phi_{23}\\\Phi_{33}\ea\right)=
\left(\ba{l}0\\0\\1\ea\right)+\frac{1}{k}\left(\ba{l}\Phi_{13}^{(1)}\\\Phi_{23}^{(1)}\\\Phi_{33}^{(1)}\ea\right)+
\frac{1}{k^2}\left(\ba{l}\Phi_{13}^{(2)}\\\Phi_{23}^{(2)}\\\Phi_{33}^{(2)}\ea\right)+O(\frac{1}{k^3})\\
{}+\left[\frac{1}{k}\left(\ba{c}-\Phi^{(1)}_{13}(0)\\0\\0\ea\right)+
\frac{1}{k^2}\left(\ba{c}-\Phi^{(2)}_{13}(0)+\Phi^{(1)}_{13}(0)\Phi^{(1)}_{23}+\Phi^{(1)}_{13}(0)\Phi^{(1)}_{33}\\-\frac{i\sig}{2}\bar g_{01}\Phi^{(1)}_{13}(0)\\-\frac{i\sig}{2}\bar g_{02}\Phi^{(1)}_{13}(0)\ea\right)+O(\frac{1}{k^2})\right]e^{-4ik^2t}
\ea
\ee

Similar to (\ref{Phiteqn}), we have the equations $\{\phi_{ij}\}_{i,j=1}^{3}$ satisfy the similar partial derivative equations:\\
\begin{subequations}
From the first column of the equation (\ref{Phit}) we get
\be\label{phi1t}
\left\{
\ba{l}
\phi_{11t}=2k(f_{01}\phi_{21}+f_{02}\phi_{31})-i\sig (|f_{01}|^2+|f_{02}|^2)\phi_{11}+i(f_{11}\phi_{21}+f_{12}\phi_{31}),\\
\phi_{21t}-4ik^2\phi_{21}=2\sig k \bar f_{01}\phi_{11}+i\sig (|f_{01}|^2\phi_{21}+\bar f_{01}f_{02}\phi_{31})-i\sig \bar f_{11}\phi_{11},\\
\phi_{31t}-4ik^2\phi_{31}=2\sig k \bar f_{02}\phi_{11}+i\sig (\bar f_{02}f_{01}\phi_{21}+|f_{02}|^2\phi_{31})-i\sig \bar f_{12}\phi_{11},
\ea
\right.
\ee

From the second column of the equation (\ref{Phit}) we get
\be\label{phi2t}
\left\{
\ba{l}
\phi_{12t}+4ik^2\phi_{12}=2k(f_{01}\phi_{22}+f_{02}\phi_{32})-i\sig (|f_{01}|^2+|f_{02}|^2)\phi_{12}+i(f_{11}\phi_{22}+f_{12}\phi_{32}),\\
\phi_{22t}=2\sig k \bar f_{01}\phi_{12}+i\sig (|f_{01}|^2\phi_{22}+\bar f_{01}f_{02}\phi_{32})-i\sig \bar f_{11}\phi_{12},\\
\phi_{32t}=2\sig k \bar f_{02}\phi_{12}+i\sig (\bar f_{02}f_{01}\phi_{22}+|f_{02}|^2\phi_{32})-i\sig \bar f_{12}\phi_{12},
\ea
\right.
\ee
From the third column of the equation (\ref{Phit}) we get
\be\label{phi3t}
\left\{
\ba{l}
\phi_{13t}+4ik^2\phi_{13}=2k(f_{01}\phi_{23}+f_{02}\phi_{33})-i\sig (|f_{01}|^2+|f_{02}|^2)\phi_{13}+i(f_{11}\phi_{23}+f_{12}\phi_{33}),\\
\phi_{23t}=2\sig k \bar f_{01}\phi_{13}+i\sig (|f_{01}|^2\phi_{23}+\bar f_{01}f_{02}\phi_{33})-i\sig \bar f_{11}\phi_{13},\\
\phi_{33t}=2\sig k \bar f_{02}\phi_{13}+i\sig (\bar f_{02}f_{01}\phi_{23}+|f_{02}|^2\phi_{33})-i\sig \bar f_{12}\phi_{13},
\ea
\right.
\ee
\end{subequations}

Then, substituting these formulas into the equation (\ref{c12def}) and noticing that we assume that the initial value and boundary value are compatible at $x=0$ and $x=L$, we get the asymptotic behavior (\ref{c12largek}) of $c_{1j}(t,k)$ as $k\rightarrow \infty$. Similar to prove the formula (\ref{c13largek}).

\end{proof}

\subsection{The Dirichlet and Neumann problems}
We can now derive effective characterizations of spectral function $S(k),S_L(k)$ for the Dirichlet ($\{g_{01}(t),g_{02}(t)\}$ and $\{f_{01}(t),f_{02}(t)\}$ prescribed), the Neumann ($\{g_{11}(t),g_{12}(t)\}$ and $\{f_{11}(t),f_{12}(t)\}$ prescribed) problems.
\par
Define functions as
\be\label{Omegadef}
f_-(t,k)=f(t,k)-f(t,-k),\quad f_+(t,k)=f(t,k)+f(t,-k),
\ee
Introducing
\be
\Dta(k)=e^{2ikL}-e^{-2ikL},\quad \Sig(k)=e^{2ikL}+e^{-2ikL}
\ee
Denoting $\pt D_3^0$ as the boundary contour which is not included the zeros of $\Dta(k)$.

\begin{theorem}\label{maintheom}
Let $T<\infty$. Let $q_{10}(x),q_{20}(x),0\le x\le L$, be two initial functions.
\par
For the Dirichlet problem it is assumed that the function $\{g_{01}(t),g_{02}(t)\},0\le t<T$, has sufficient smoothness and is compatible with $q_{10}(x),q_{20}(x)$ at $x=t=0$, that is
\[
q_{10}(0)=g_{01}(0),\quad q_{20}(0)=g_{02}(0).
\]
The function $\{f_{01}(t),f_{02}(t)\},0\le t<T$, has sufficient smoothness and is compatible with $q_{10}(x),q_{20}(x)$ at $x=L$, that is,
\[
q_{10}(L)=f_{01}(0),\quad q_{20}(L)=f_{02}(0).
\]
\par
For the Neumann problem it is assumed that the function $g_1(t),0\le t<T$, has sufficient smoothness and is compatible with $q_0(x)$ at $x=t=0$; the function $f_1(t),0\le t<T$, has sufficient smoothness and is compatible with $q_0(x)$ at $x=L$.

\par
For simplicity, we suppose that $m_{11}(\mathcal{A})(k)$ has no zeros in $D_1$.
\par
Then the spectral function $S(k)$ is given by
\be\label{SK}
S(k)=\left(\ba{ccc}
\ol{\Phi_{11}(\bar k)}&-\sig \ol{\Phi_{21}(\bar k)}e^{4ik^2T}&-\sig \ol{\Phi_{31}(\bar k)}e^{4ik^2T}\\
-\sig \ol{\Phi_{12}(\bar k)}e^{-4ik^2T}&\ol{\Phi_{22}(\bar k)}&\ol{\Phi_{32}(\bar k)}\\
-\sig \ol{\Phi_{13}(\bar k)}e^{-4ik^2T}&\ol{\Phi_{23}(\bar k)}&\ol{\Phi_{33}(\bar k)}\ea\right)
\ee
\be\label{SLk}
S_L(k)=\left(\ba{ccc}
\ol{\phi_{11}(\bar k)}&-\sig \ol{\phi_{21}(\bar k)}e^{4ik^2T}&-\sig \ol{\phi_{31}(\bar k)}e^{4ik^2T}\\
-\sig \ol{\phi_{12}(\bar k)}e^{-4ik^2T}&\ol{\phi_{22}(\bar k)}&\ol{\phi_{32}(\bar k)}\\
-\sig \ol{\phi_{13}(\bar k)}e^{-4ik^2T}&\ol{\phi_{23}(\bar k)}&\ol{\phi_{33}(\bar k)}\ea\right)
\ee
and the complex-value functions $\{\Phi_{l3}(t,k)\}_{l=1}^{3}$ satisfy the following system of integral equations:
\begin{subequations}\label{Phi3sys}
\be\label{Phi13sys}
\ba{rl}
\Phi_{13}(t,k)&=\int_0^te^{-4ik^2(t-t')}\left[-i\sig(|g_{01}|^2+|g_{02}|^2)\Phi_{13}\right.\\
&\left.+(2kg_{01}+ig_{11})\Phi_{23}+(2kg_{02}+ig_{12})\Phi_{33}\right](t',k)dt'
\ea
\ee
\be\label{Phi23sys}
\Phi_{23}(t,k)=\int_0^t\sig\left[(2k\bar g_{01}-i\bar g_{11})\Phi_{13}+i|g_{01}|^2\Phi_{23}+i\bar g_{01} g_{02}\Phi_{33}\right](t',k)dt'
\ee
\be\label{Phi33sys}
\Phi_{33}(t,k)=1+\int_0^t\sig\left[(2k\bar g_{02}-i\bar g_{12})\Phi_{13}+ig_{01}\bar g_{02}\Phi_{23}+i|g_{02}|^2\Phi_{33}\right](t',k)dt'
\ee
\end{subequations}
and $\{\Phi_{l1}(t,k)\}_{l=1}^{3},\{\Phi_{l2}(t,k)\}_{l=1}^{3}$ satisfy the following system of integral equations:
\begin{subequations}\label{Phi1sys}
\be\label{Phi11sys}
\ba{rl}
\Phi_{11}(t,k)&=1+\int_0^t\left[-i\sig(|g_{01}|^2+|g_{02}|^2)\Phi_{11}\right.\\
&\left.+(2kg_{01}+ig_{11})\Phi_{21}+(2kg_{02}+ig_{12})\Phi_{31}\right](t',k)dt'
\ea
\ee
\be\label{Phi21sys}
\Phi_{21}(t,k)=\int_0^te^{4ik^2(t-t')}\sig\left[(2k\bar g_{01}-i\bar g_{11})\Phi_{11}+i|g_{01}|^2\Phi_{21}+i\bar g_{01} g_{02}\Phi_{31}\right](t',k)dt'
\ee
\be\label{Phi33sys}
\Phi_{31}(t,k)=\int_0^te^{4ik^2(t-t')}\sig\left[(2k\bar g_{02}-i\bar g_{12})\Phi_{11}+ig_{01}\bar g_{02}\Phi_{21}+i|g_{02}|^2\Phi_{31}\right](t',k)dt'
\ee
\end{subequations}
\begin{subequations}\label{Phi2sys}
\be\label{Phi12sys}
\ba{rl}
\Phi_{12}(t,k)&=\int_0^te^{-4ik^2(t-t')}\left[-i\sig(|g_{01}|^2+|g_{02}|^2)\Phi_{12}\right.\\
&\left.+(2kg_{01}+ig_{11})\Phi_{22}+(2kg_{02}+ig_{12})\Phi_{32}\right](t',k)dt'
\ea
\ee
\be\label{Phi22sys}
\Phi_{22}(t,k)=1+\int_0^t\sig\left[(2k\bar g_{01}-i\bar g_{11})\Phi_{12}+i|g_{01}|^2\Phi_{22}+i\bar g_{01} g_{02}\Phi_{32}\right](t',k)dt'
\ee
\be\label{Phi32sys}
\Phi_{32}(t,k)=\int_0^t\sig\left[(2k\bar g_{02}-i\bar g_{12})\Phi_{12}+ig_{01}\bar g_{02}\Phi_{22}+i|g_{02}|^2\Phi_{32}\right](t',k)dt'
\ee
\end{subequations}
Functions $\{\phi_{ij}(t,k)\}_{i,j=1}^{3}$ satisfy the same integral equations replaying the functions $\{g_{01},g_{02},g_{11},g_{12}\}$ with $\{f_{01},f_{02},f_{11},f_{12}\}$.
\begin{enumerate}
\item For the Dirichlet problem, the unknown Neumann boundary value $\{g_{11}(t),g_{12}(t)\}$ and $\{f_{11}(t),f_{12}(t)\}$ are given by
\begin{subequations}\label{DtoNg1}
\be\label{DtoNg11}
\ba{rl}
g_{11}(t)&=\frac{2}{i\pi}\int_{\pt D^0_3}\frac{\Sig}{\Dta}(k\Phi_{12-}+ig_{01})dk+\frac{2}{\pi}\int_{\pt D^0_3}(g_{01}\Phi_{22-}+g_{02}\Phi_{32-})dk\\
&-\frac{1}{\pi}\int_{\pt D^0_3}(g_{01}\bar\phi_{22-}+g_{02}\bar\phi_{23-})dk-\frac{4\sig}{i\pi}\int_{\pt D^0_3}\frac{1}{\Dta}(k\bar\phi_{21-}+i\sig\bar f_{01})dk\\
&+\frac{4}{i\pi}\int_{\pt D^0_3}\frac{k}{\Dta}\left[\left(\Phi_{12}(\bar\phi_{22}-1)+\Phi_{13}\bar\phi_{23}\right)e^{-2ikL}+\sig(\Phi_{11}-1)\bar\phi_{21}\right]_-dk.
\ea
\ee
\be\label{DtoNg12}
\ba{rl}
g_{12}(t)&=\frac{2}{i\pi}\int_{\pt D^0_3}\frac{\Sig}{\Dta}(k\Phi_{13-}+ig_{02})dk+\frac{2}{\pi}\int_{\pt D^0_3}(g_{01}\Phi_{23-}+g_{02}\Phi_{33-})dk\\
&-\frac{1}{\pi}\int_{\pt D^0_3}(g_{01}\bar\phi_{32-}+g_{02}\bar\phi_{33-})dk-\frac{4\sig}{i\pi}\int_{\pt D^0_3}\frac{1}{\Dta}(k\bar\phi_{31-}+i\sig\bar f_{02})dk\\
&+\frac{4}{i\pi}\int_{\pt D^0_3}\frac{k}{\Dta}\left[\left(\Phi_{12}\bar\phi_{32}+\Phi_{13}(\bar\phi_{33}-1)\right)e^{-2ikL}+\sig(\Phi_{11}-1)\bar\phi_{31}\right]_-dk.
\ea
\ee
\end{subequations}
and
\begin{subequations}\label{DtoNf1}
\be\label{DtoNf11}
\ba{rl}
f_{11}(t)&=-\frac{2}{i\pi}\int_{\pt D^0_3}\frac{\Sig}{\Dta}(k\phi_{12-}+if_{01})dk+\frac{2}{\pi}\int_{\pt D^0_3}(f_{01}\phi_{22-}+f_{02}\phi_{32-})dk\\
&+\frac{1}{\pi}\int_{\pt D^0_3}(f_{01}\bar\Phi_{22-}+f_{02}\bar\Phi_{23-})dk+\frac{4\sig}{i\pi}\int_{\pt D^0_3}\frac{1}{\Dta}(k\bar\Phi_{21-}+i\sig\bar g_{01})dk\\
&+\frac{4}{i\pi}\int_{\pt D^0_3}\frac{k}{\Dta}\left[\sig(\phi_{11}-1)\bar\Phi_{21}-\left(\phi_{12}(\bar\Phi_{22}-1)+\phi_{13}\bar\Phi_{23}\right)e^{2ikL}\right]_-dk.
\ea
\ee
\be\label{DtoNf12}
\ba{rl}
f_{12}(t)&=-\frac{2}{i\pi}\int_{\pt D^0_3}\frac{\Sig}{\Dta}(k\phi_{13-}+if_{02})dk+\frac{2}{\pi}\int_{\pt D^0_3}(f_{03}\phi_{23-}+f_{02}\phi_{33-})dk\\
&+\frac{1}{\pi}\int_{\pt D^0_3}(f_{01}\bar\Phi_{32-}+f_{02}\bar\Phi_{33-})dk+\frac{4\sig}{i\pi}\int_{\pt D^0_3}\frac{1}{\Dta}(k\bar\Phi_{31-}+i\sig\bar g_{02})dk\\
&+\frac{4}{i\pi}\int_{\pt D^0_3}\frac{k}{\Dta}\left[\sig(\phi_{11}-1)\bar\Phi_{31}-\left(\phi_{12}\bar\Phi_{32}+\phi_{13}(\bar\Phi_{33}-1)\right)e^{2ikL}\right]_-dk.
\ea
\ee
where the conjugate of a function $h$ denotes $\bar h=\bar h(\bar k)$.
\end{subequations}
\item For the Neumann problem, the unknown boundary values $\{g_{01}(t),g_{02}(t)\}$ and $\{f_{01}(t),f_{02}(t)\}$ are given by
\begin{subequations}
\be\label{NtoDg01}
\ba{rl}
g_{01}(t)&=\frac{1}{\pi}\int_{\pt D^0_3}\frac{\Sig}{\Dta}\Phi_{12+}dk-\frac{2}{\pi}\int_{\pt D^0_3}\frac{\sig}{\Dta}\bar \phi_{21+}dk\\
&-\frac{2}{\pi}\int_{\pt D^0_3}\frac{1}{\Dta}\left[\sig(\Phi_{11}-1)\bar\phi_{21}-(\Phi_{12}(\bar\phi_{22}-1)+\Phi_{13}\bar\phi_{23})e^{-2ikL}\right]_+dk,
\ea
\ee
\be\label{NtoDg02}
\ba{rl}
g_{02}(t)&=\frac{1}{\pi}\int_{\pt D^0_3}\frac{\Sig}{\Dta}\Phi_{13+}dk-\frac{2}{\pi}\int_{\pt D^0_3}\frac{\sig}{\Dta}\bar \phi_{31+}dk\\
&-\frac{2}{\pi}\int_{\pt D^0_3}\frac{1}{\Dta}\left[\sig(\Phi_{11}-1)\bar\phi_{31}-(\Phi_{12}\bar\phi_{32}+\Phi_{13}(\bar\phi_{33}-1))e^{-2ikL}\right]_+dk,
\ea
\ee
\end{subequations}
and
\begin{subequations}
\be\label{NtoDf01}
\ba{rl}
f_{01}(t)&=-\frac{1}{\pi}\int_{\pt D^0_3}\frac{\Sig}{\Dta}\phi_{12+}dk-\frac{2}{\pi}\int_{\pt D^0_3}\frac{\sig}{\Dta}\bar \Phi_{21+}dk\\
&-\frac{2}{\pi}\int_{\pt D^0_3}\frac{1}{\Dta}\left[\sig(\phi_{11}-1)\bar\Phi_{21}-(\phi_{12}(\bar\Phi_{22}-1)+\phi_{13}\bar\Phi_{23})e^{2ikL}\right]_+dk,
\ea
\ee
\be\label{NtoDf02}
\ba{rl}
f_{02}(t)&=-\frac{1}{\pi}\int_{\pt D^0_3}\frac{\Sig}{\Dta}\phi_{13+}dk-\frac{2}{\pi}\int_{\pt D^0_3}\frac{\sig}{\Dta}\bar \Phi_{31+}dk\\
&-\frac{2}{\pi}\int_{\pt D^0_3}\frac{1}{\Dta}\left[\sig(\phi_{11}-1)\bar\Phi_{31}-(\phi_{12}\bar\Phi_{32}+\phi_{13}(\bar\Phi_{33}-1))e^{2ikL}\right]_+dk,
\ea
\ee
\end{subequations}

\end{enumerate}
\end{theorem}
\begin{proof}
The representations (\ref{SK}) and (\ref{SLk}) follow from the relation $S(k)=e^{-2ik^2t\hat \Lam}\mu_2^{-1}(0,t,k)$ and $S_L(k)=e^{-2ik^2t\hat\Lam}\mu_3^{-1}(0,t,k)$, respectively. And the system (\ref{Phil3sys}) is the direct result of the Volteral integral equations of $\mu_2(0,t,k)$ and $\mu_3(L,t,k)$.
\begin{enumerate}
\item In order to derive (\ref{DtoNg11}) we note that equation (\ref{g1}) expresses $g_{11}(t)$ 
in terms of $\Phi_{22}^{(1)}(t,k),\Phi_{32}^{(1)}(t,k),\Phi^{(2)}_{12}(t,k)$. 
Furthermore,  equation (\ref{mu2x0tk}) and Cauchy theorem imply
    \[
    \ba{l}
    -\frac{\pi i}{2}\Phi_{22}^{(1)}(t)=\int_{\pt D_2}[\Phi_{22}(t,k)-1]dk=\int_{\pt D_4}[\Phi_{2\times 2}(t,k)-1]dk,\\
    -\frac{\pi i}{2}\Phi_{32}^{(1)}(t)=\int_{\pt D_2}\Phi_{32}(t,k)dk=\int_{\pt D_4}\Phi_{32}(t,k)dk,
    \ea
    \]
    and
    \[
    -\frac{\pi i}{2}\Phi_{12}^{(2)}(t)=\int_{\pt D_2}\left[k\Phi_{12}(t,k)-\frac{g_{01}(t)}{2i}\right]dk=\int_{\pt D_4}\left[k\Phi_{12}(t,k)-\frac{g_{01}(t)}{2i}\right]dk,\\
    \]
    Thus,
    \be\label{Phi331}
    \ba{rl}
    i\pi\Phi_{22}^{(1)}(t)&=-\left(\int_{\pt D_2}+\int_{\pt D_4}\right)[\Phi_{22}(t,k)-1]dk\\
    &=\left(\int_{\pt D_1}+\int_{\pt D_3}\right)[\Phi_{22}(t,k)-1]dk\\
    &=\int_{\pt D_3}[\Phi_{22}(t,k)-1]dk-\int_{\pt D_3}[\Phi_{22}(t,-k)-1]dk\\
    &=\int_{\pt D_3}(\Phi_{22}(t,k)-\Phi_{22}(t,-k))dk\\
    &=\int_{\pt D_3}\Phi_{22-}(t,k)dk\\
    i\pi\Phi_{32}^{(1)}(t)&=-\left(\int_{\pt D_2}+\int_{\pt D_4}\right)[\Phi_{32}(t,k)]dk\\
    &=\int_{\pt D_3}\Phi_{32-}(t,k)dk
    \ea
    \ee
    Similarly,
    \be\label{Phi132}
    \ba{rl}
    i\pi \Phi_{12}^{(2)}(t)&=\left(\int_{\pt D_3}+\int_{\pt D_1}\right)\left[k\Phi_{12}(t,k)-\frac{g_{01}(t)}{2i}\right]dk\\
    &=\int_{\pt D_3}\left[k\Phi_{12}(t,k)-\frac{g_{01}(t)}{2i}\right]_-dk\\
    &=\int_{\pt D^0_3}\left\{k\Phi_{12}(t,k)-\frac{g_{01}(t)}{2i}+\frac{2e^{-2ikL}}{\Dta}[k\Phi_{12}(t,k)-\frac{g_{01}(t)}{2i}]\right\}_-dk+I(t)
    \ea
    \ee
    where $I(t)$ is defined by
    \[
    I(t)=-\int_{\pt D^0_3}\left\{\frac{2e^{-2ikL}}{\Dta}[k\Phi_{12}(t,k)-\frac{g_{01}(t)}{2i}]\right\}_-dk
    \]
    The last step involves using the global relation (\ref{vanishglo}) to compute $I(t)$. That is,
    \be\label{DtoNIt}
    \ba{rl}
    I(t)&=\int_{\pt D^0_3}\left\{-\frac{2e^{-2ikL}}{\Dta}\left[kc_{12}-\Phi^{(1)}_{12}-\frac{\Phi^{(1)}_{12}\bar \phi^{(1)}_{22}+\Phi^{(1)}_{13}\bar \phi^{(1)}_{23}}{k}+\sig \bar \phi^{(1)}_{21}e^{2ikL}\right]\right\}_-dk\\
    &+\int_{\pt D^0_3}\left\{-\frac{2e^{-2ikL}}{\Dta}\left[\frac{\Phi^{(1)}_{12}\bar \phi^{(1)}_{22}+\Phi^{(1)}_{13}\bar \phi^{(1)}_{23}}{k}+\sig (k\bar\phi_{21}-\bar \phi^{(1)}_{21})e^{2ikL}\right]\right\}_-dk\\
    &+\int_{\pt D^0_3}\left\{\frac{2e^{-2ikL}}{\Dta}\left[k\left(\Phi_{12}(\bar\phi_{22}-1)+\Phi_{13}\bar\phi_{23}+\sig(\Phi_{11}-1)\bar\phi_{21}e^{2ikL}\right)\right]\right\}_-dk
    \ea
    \ee
    Using the asymptotic (\ref{cjlargek}) and Cauchy theorem to compute these terms on the right-hand side of equation (\ref{DtoNIt}), we find
    \be\label{DtoNItres}
    \ba{rl}
    I(t)&=-i\pi \Phi^{(2)}_{12}-\int_{\pt D^0_3}\left[\frac{g_{01}}{2i}\bar\phi_{22-}-\frac{2\sig}{\Dta}(k\bar\phi_{21-}+i\sig f_{01})\right]dk\\
    &+\int_{\pt D^0_3}\frac{2k}{\Dta}\left[\left(\Phi_{12}(\bar\phi_{22}-1)+\Phi_{13}\bar\phi_{23}\right)e^{-2ikL}+\sig(\Phi_{11}-1)\bar\phi_{21}\right]_-dk.
    \ea
    \ee
    Equations (\ref{Phi331}), (\ref{Phi132}) and (\ref{DtoNItres}) together with (\ref{g1}) yield (\ref{DtoNg11}). Similarly, we can prove (\ref{DtoNg12}).

    \par
    The expression (\ref{DtoNf11}) for $f_{11}(t)$ can be derived in a similar way. Indeed, we note that equation (\ref{f1}) expresses $f_{11}(t)$ in terms of $\phi_{22}^{(1)}(t,k),\phi_{32}^{(1)}(t,k),\phi^{(2)}_{12}(t,k)$. These three functions satisfy the analog of equations (\ref{Phi331}) and $\ref{Phi132}$. In particular, $\phi^{(2)}_{12}$ satisifies
    \be\label{phi122}
    \ba{rl}
    i\pi\phi^{(2)}_{12}(t)&=\int_{\pt D_3}\left[k\phi_{12}-\phi^{(1)}_{12}\right]_-dk\\
    &=\int_{\pt D^0_3}\left\{-\frac{\Sig}{\Dta}(k\phi_{12}-\phi^{(1)}_{12})\right\}_-dk+J(t)
    \ea
    \ee
    where
    \be
    J(t)=\int_{\pt D^0_3}\left\{\frac{2e^{2ikL}}{\Dta}(k\phi_{12}-\phi^{(1)}_{12})\right\}_-dk
    \ee
    Then using the global relation (\ref{vanishglo}) to compute $J(t)$:
    \be\label{Jt}
    \ba{rl}
    J(t)&=\int_{\pt D^0_3}\left\{-\frac{2}{\Dta}\left[k\sig\bar c_{21}+\phi^{(1)}_{12}e^{2ikL}-\sig\bar\Phi^{(1)}_{21}+\frac{\phi^{(1)}_{12}\bar\Phi_{22}+\phi_{13}\bar\Phi_{23}}{k}e^{2ikL}\right]\right\}_-dk\\
    &+\int_{\pt D^0_3}\frac{2}{\Dta}\left\{\frac{\phi^{(1)}_{12}\bar\Phi_{22}+\phi_{13}\bar\Phi_{23}}{k}e^{2ikL}+\sig(k\bar\Phi_{21}-\bar\Phi^{(1)}_{21})\right\}_-dk\\
    &+\int_{\pt D^0_3}\frac{2k}{\Dta}\left[\sig(\phi_{11}-1)\bar\Phi_{21}-(\phi_{12}(\bar\Phi_{22}-1)+\phi_{13}\bar\phi_{23})e^{2ikL}\right]_-dk
    \ea
    \ee
    The equations (\ref{Jt}), (\ref{phi122}) together with the asymptotics of $c_{21}(t,k)$ yield (\ref{DtoNf11}). The proof of (\ref{DtoNf12}) is similar.

    \bigskip

    \item In order to derive the representations (\ref{NtoDg01}) relevant for the Neumann problem, we note that equation (\ref{g0}) expresses $g_{01}$ and $g_{02}$ in terms of $\Phi^{(1)}_{12}$ and $\Phi^{(1)}_{13}$, respectively. Furthermore, equation (\ref{mu2x0tk}) and Cauchy theorem imply

        \be\label{Phi131}
        \ba{l}
        -\frac{\pi i}{2}\Phi^{(1)}_{12}(t)=\int_{\pt D_2}\Phi_{12}(t,k)dk=\int_{\pt D_4}\Phi_{12}(t,k)dk,
        \ea
        \ee
        Thus,
        \be\label{Phi1j1}
        \ba{l}
        i\pi \Phi^{(1)}_{12}(t)=\left(\int_{\pt D_3}+\int_{\pt D_1}\right)\Phi_{12}(t,k)dk\\
        {}=\int_{\pt D_3}\Phi_{12-}(t,k)dk\\
        {}=\int_{\pt D^0_3}\left(\frac{\Sig}{\Dta}\Phi_{12+}(t,k)\right)dk+K(t),
        \ea
        \ee
        where
        \be
        K(t)=-\int_{\pt D^0_3}\frac{2}{\Dta}\left(e^{-2ikL}\Phi_{1j}\right)_+dk
        \ee
        and using the global relation and the asymptotic formulas of $c_{12}$, we have
        \be\label{Kt}
        \ba{l}
        K(t)=-i\pi\Phi^{(1)}_{12}-\int_{\pt D^0_3}\left\{\frac{2\sig}{\Dta}\bar\phi_{21+}+\frac{2}{\Dta}\left[\sig(\Phi_{11}-1)\bar\phi_{21}-\left(\Phi_{12}(\bar\phi_{22}-1)+\Phi_{13}\bar\phi_{23}e^{-2ikL}\right)\right]_+\right\}dk
        \ea
        \ee
        Equations (\ref{g0}), (\ref{Phi1j1}) and (\ref{Kt}) yields (\ref{NtoDg01}). The proof of the other formulas is similar.
\end{enumerate}
\end{proof}

\subsection{Effective characterizations}
Substituting into the system (\ref{Phil3sys}) the expressions
\begin{subequations}
\be\label{Phij3eps}
\Phi_{ij}=\Phi_{ij,0}+\eps\Phi_{ij,1}+\eps^2\Phi_{ij,2}+\cdots,\quad i,j=1,2,3.
\ee
\be\label{phij3eps}
\phi_{ij}=\phi_{ij,0}+\eps\phi_{ij,1}+\eps^2\phi_{ij,2}+\cdots,\quad i,j=1,2,3.
\ee
\be\label{g0eps}
g_{01}=\eps g^{(1)}_{01}+\eps^2 g^{(2)}_{01}+\cdots,\quad g_{02}=\eps g^{(1)}_{02}+\eps^2 g^{(2)}_{02}+\cdots,
\ee
\be\label{f0eps}
f_{01}=\eps f^{(1)}_{01}+\eps^2 f^{(2)}_{01}+\cdots,\quad f_{02}=\eps f^{(1)}_{02}+\eps^2 f^{(2)}_{02}+\cdots,
\ee
\be\label{g1eps}
g_{11}=\eps g^{(1)}_{11}+\eps^2 g^{(2)}_{11}+\cdots,\quad g_{12}=\eps g^{(1)}_{12}+\eps^2 g^{(2)}_{12}+\cdots,
\ee
\be\label{g1eps}
f_{11}=\eps f^{(1)}_{11}+\eps^2 f^{(2)}_{11}+\cdots,\quad f_{12}=\eps f^{(1)}_{12}+\eps^2 f^{(2)}_{12}+\cdots,
\ee
\end{subequations}
where $\eps>0$ is a small parameter, we find that the terms of $O(1)$ give
\be\label{Oeps0}
O(1):\left\{
\ba{ccc}
\Phi_{13,0}=0 & \Phi_{23,0}=0 & \Phi_{33,0}=1,\\
\Phi_{11,0}=1 & \Phi_{21,0}=0 & \Phi_{31,0}=0,\\
\Phi_{12,0}=0 & \Phi_{22,0}=1 & \Phi_{32,0}=0.
\ea
\right.
\ee
Moreover, the terms of $O(\eps)$ give
\be\label{Oeps}
O(\eps):\left\{
\ba{l}
\Phi_{33,1}=0 \quad \Phi_{23,1}=0,\\
\Phi_{13,1}(t,k)=\int_0^te^{-4ik^2(t-t')}(2kg^{(1)}_{02}+ig^{(1)}_{12})(t')dt',\\
\Phi_{11,1}=0,\\
\Phi_{21,1}=\int_{0}^{t}\sig e^{4ik^2(t-t')}(2k\bar g^{(1)}_{01}-i\bar g^{(1)}_{11})(t')dt',\\
\Phi_{31,1}=\int_{0}^{t}\sig e^{4ik^2(t-t')}(2k\bar g^{(1)}_{02}-i\bar g^{(1)}_{12})(t')dt',\\
\Phi_{12,1}=\int_{0}^{t}e^{-4ik^2(t-t')}(2kg^{(1)}_{01}+ig^{(1)}_{11})(t')dt',\\
\Phi_{22,1}=0,\quad \Phi_{32,1}=0.
\ea
\right.
\ee
and the terms of $O(\eps^2)$ give
\be\label{Oeps2}
O(\eps^2):\left\{
\ba{l}
\Phi_{13,2}=\int_{0}^{t}e^{-4ik^2(t-t')}(2kg^{(2)}_{02}+ig^{(2)}_{12})(t')dt',\\
\Phi_{23,2}=\int_{0}^{t}\sig[(2k\bar g^{(1)}_{01}-i\bar g^{(1)}_{11})(t')\Phi_{13,1}(t',k)+i\bar g^{(1)}_{01}(t')g^{(1)}_{02}(t')]dt',\\
\Phi_{33,2}=\int_{0}^{t}\sig[(2k\bar g^{(1)}_{02}-i\bar g^{(1)}_{12})(t')\Phi_{13,1}(t',k)+i|g^{(1)}_{02}(t')|^2]dt',\\
\Phi_{11,2}=\int_{0}^{t}\left[-i\sig(|g^{(2)}_{01}|^2+|g^{(2)}_{02}|^2)(t')+(2kg^{(1)}_{01}+ig^{(1)}_{11})(t')\Phi_{21,1}(t',k)\right.\\
{}\left.+(2kg^{(1)}_{02}+ig^{(1)}_{12})(t')\Phi_{31,1}(t',k)\right]dt',\\
\Phi_{21,2}=\int_{0}^{t}\sig e^{4ik^2(t-t')}(2k\bar g^{(2)}_{01}-i\bar g^{(2)}_{11})(t')dt',\\
\Phi_{31,2}=\int_{0}^{t}\sig e^{4ik^2(t-t')}(2k\bar g^{(2)}_{02}-i\bar g^{(2)}_{12})(t')dt',\\
\Phi_{12,2}=\int_{0}^{t}e^{4ik^2(t-t')}(2kg^{(2)}_{01}+ig^{(2)}_{11})(t')dt',\\
\Phi_{22,2}=\int_{0}^{t}\sig\left[(2k\bar g^{(1)}_{01}-i\bar g^{(1)}_{11})(t')\Phi_{12,1}(t',k)+i|g^{(1)}_{01}(t')|^2\right]dt',\\
\Phi_{32,2}=\int_{0}^{t}\sig\left[(2k\bar g^{(1)}_{02}-i\bar g^{(1)}_{12})(t')\Phi_{12,1}(t',k)+ig^{(1)}_{01}(t')\bar g^{(1)}_{02}(t')\right]dt'.
\ea
\right.
\ee

Similarly, we will have the analogue formulas for $\{\phi_{ij,l}\}_{i,j=1}^{3},l=0,1,2$ expressed in terms of the boundary data at $x=L$, that is $\{f^{(l)}_{ij}\}_{i=0,1}^{j=1,2},l=1,2$.
\par
On the other hand, expanding (\ref{DtoNg1}) and assuming for simplicity that $m_{11}(\mathcal{A})(k)$ has no zeros, we find
\begin{subequations}\label{DtoNgf^1}
\be\label{DtoNg^11}
g^{(1)}_{11}(t)=\frac{2}{\pi i}\int_{\pt D^0_3}(k\Phi_{12,1-}(t,k)+ig^{(1)}_{01})dk-\frac{4\sig}{\pi i}\int_{\pt D^0_3}\frac{1}{\Dta}(k\bar\phi_{21,1-}+i\sig\bar f^{(1)}_{01})dk,
\ee
\be\label{DtoNg^12}
g^{(1)}_{12}(t)=\frac{2}{\pi i}\int_{\pt D^0_3}(k\Phi_{13,1-}(t,k)+ig^{(1)}_{02})dk-\frac{4\sig}{\pi i}\int_{\pt D^0_3}\frac{1}{\Dta}(k\bar\phi_{31,1-}+i\sig\bar f^{(1)}_{02})dk,
\ee
\be\label{DtoNg^11}
f^{(1)}_{11}(t)=-\frac{2}{\pi i}\int_{\pt D^0_3}(k\phi_{12,1-}(t,k)+if^{(1)}_{01})dk+\frac{4\sig}{\pi i}\int_{\pt D^0_3}\frac{1}{\Dta}(k\bar\Phi_{21,1-}+i\sig\bar g^{(1)}_{01})dk,
\ee
\be\label{DtoNg^12}
f^{(1)}_{12}(t)=-\frac{2}{\pi i}\int_{\pt D^0_3}(k\phi_{13,1-}(t,k)+if^{(1)}_{02})dk+\frac{4\sig}{\pi i}\int_{\pt D^0_3}\frac{1}{\Dta}(k\bar\Phi_{31,1-}+i\sig\bar g^{(1)}_{02})dk,
\ee
\end{subequations}
\par
We also find that
\be\label{Om^1}
\ba{l}
\Phi_{12,1-}=4k\int_{0}^{t}e^{-4ik^2(t-t')}g^{(1)}_{01}(t')dt',\\
\Phi_{13,1-}=4k\int_{0}^{t}e^{-4ik^2(t-t')}g^{(1)}_{02}(t')dt',\\
\phi_{21,1-}=4\sig k\int_{0}^{t}e^{4ik^2(t-t')}\bar f^{(1)}_{01}(t')dt',\\
\phi_{31,1-}=4\sig k\int_{0}^{t}e^{4ik^2(t-t')}\bar f^{(1)}_{02}(t')dt'.\\
\ea
\ee
The Dirichlet problem can now be solved perturbatively as follows: assuming for simplicity that $m_{11}(\mathcal{A})(k)$ has no zeros and given $g^{(1)}_{01},g^{(1)}_{02}$ and $f^{(1)}_{01},f^{(1)}_{02}$, we can use equation (\ref{Om^1}) to determine $\Om^{(1)}$. We can then compute $g^{(1)}_{11},g^{(1)}_{12}$ and $f^{(1)}_{11},f^{(1)}_{12}$ from (\ref{DtoNg^1}) and then $\Phi_{1j,1},j=2,3$ from (\ref{Oeps}) and the analogue results for $\phi_{j1,1},j=2,3$. In the same way we can determine $\Phi_{1j,2},j=2,3$ from (\ref{Oeps2}) and the analogue results for $\phi_{j1,2},j=2,3$, then compute $g^{(2)}_{11},g^{(2)}_{12}$ and $f^{(2)}_{11},f^{(2)}_{12}$. And these arguments can be extended to the higher order and also can be extended to the systems (\ref{Phi1sys}), (\ref{Phi2sys}) and (\ref{Phi3sys}), thus yields a constructive scheme for computing $S(k)$ to all orders. The construction of $S_L(k)$ is similar.
\par
Similarly, these arguments also can be used to the Neumann problem. That is to say, in all cases, the system can be solved perturbatively to all orders.

\subsection{The large $L$ limit}

In the limit $L\rightarrow \infty$, the representations for $g_{11}(t),g_{12}(t)$ and $g_{01}(t),g_{02}(t)$ of theorem \label{maintheom} reduce to the corresponding representations on the half-line. Indeed, as $L\rightarrow \infty$,
\[
\ba{l}
f_{01}\rightarrow 0,\quad f_{02}\rightarrow 0,\quad f_{11}\rightarrow 0,\quad f_{12}\rightarrow 0,\\
\phi_{ij}\rightarrow \dta_{ij},\quad \frac{\Sig}{\Dta}\rightarrow 1\mbox{ as $k\rightarrow \infty$ in $D_3$}
\ea
\]
Thus, the $L\rightarrow \infty$ limits of the representations (\ref{DtoNg11}), (\ref{DtoNg12}) and (\ref{NtoDg01}), (\ref{NtoDg02}) are
\be
\ba{rl}
g_{11}(t)&=\frac{2}{i\pi}\int_{\pt D^0_3}(k\Phi_{12-}+ig_{01})dk+\frac{2}{\pi}\int_{\pt D^0_3}(g_{01}\Phi_{22-}+g_{02}\Phi_{32-})dk-\frac{1}{\pi}\int_{\pt D^0_3}g_{01}\bar\phi_{22-}dk\\
g_{12}(t)&=\frac{2}{i\pi}\int_{\pt D^0_3}(k\Phi_{13-}+ig_{02})dk+\frac{2}{\pi}\int_{\pt D^0_3}(g_{01}\Phi_{23-}+g_{02}\Phi_{33-})dk-\frac{1}{\pi}\int_{\pt D^0_3}g_{02}\bar\phi_{33-}dk.
\ea
\ee
and
\be
\ba{ll}
g_{01}(t)=\frac{1}{\pi}\int_{\pt D^0_3}\Phi_{12+}dk,&
g_{02}(t)=\frac{1}{\pi}\int_{\pt D^0_3}\Phi_{13+}dk,
\ea
\ee
respectively, and these formulas coincide with the corresponding half-line formulas, see \cite{jf3}.
\bigskip

{\bf Acknowledgements}
This work of Xu was supported by National
Science Foundation of China under project NO.11501365, Shanghai Sailing Program
supported by Science and Technology Commission of Shanghai Municipality
under Grant NO.15YF1408100, Shanghai youth teacher assistance program NO.ZZslg15056 and the Hujiang Foundation of China (B14005). Fan was support by grants from the National
Science Foundation of China (Project No.10971031; 11271079; 11075055).

\end{document}